\documentclass[11pt]{amsart}
\usepackage{amsfonts}
\usepackage{amsmath}
\usepackage{amssymb}
\usepackage{graphicx}%
\usepackage{dcolumn}% Align table columns on decimal point
\usepackage{bm}
\usepackage[all]{xy}% Para los diagramas.
\input {diagxy} %\input xy\usepackage{xypic} %\xyoption{curve} \xyoption{line}\xyoption{matrix}

\def\dim{\hbox{{\rm dim}}}              %%several definitions

\def\lie#1{\mathcal{L}_{#1}}
\def\spec{\hbox{{\rm Spec}}}

\def\reduction{{/\hskip-2.5pt/\hskip-2pt}}

\def\R{{\mathbb{R}}}

\newtheorem{theorem}{Theorem}

\newtheorem{definition}{Definition}
\newtheorem{lemma}{Lemma}

\newcommand{\noi}{\noindent}

\begin{document}

\title[Nilpotent integrability and reduction]{Nilpotent integrability, reduction of dynamical systems and a third-order Calogero-Moser system}

\author{A. Ibort$^{1,2}$}\address{$^1$Depto. de Matem\'aticas, Univ. Carlos III de
Madrid, Avda. de la Universidad 30, 28911 Legan\'es, Madrid, Spain}
\address{$^2$Instituto de Ciencias Matem\'aticas, ICMAT, C/ Nicol\'as Cabrera, No 13--15, 28049 Madrid, Spain.}
%\email{albertoi@math.uc3m.es}

\author{G. Marmo$^3$}\address{$^3$Dipartimento di Fisica dell' Universit\`a ``Federico II" di Napoli,  Sezione INFN di Napoli, Complesso Universitario di Monte S. Angelo, via Cintia, 80126 Naples, Italy}
%\email{marmo@na.infn.it}

\author{M. A. Rodr\'iguez$^4$}\address{$^3$Depto. de F\'isica Te\'orica II, Univ. Complutense de Madrid,
Avda. Complutense, s/n, 28040 Madrid, Spain.}
%\email{rodrigue@fis.ucm.es}

\author{P. Tempesta$^{2,4}$}

\email{albertoi@math.uc3m.es, marmo@na.infn.it, rodrigue@fis.ucm.es
p.tempesta@fis.ucm.es, piergiulio.tempesta@icmat.es}

%\textbf{Acknowledgments}.

\thanks{AI would like to thank the partial support by the Spanish MICIN grant
MTM 2010-21186-C02-02 and QUITEMAD+. GM wishes to thank the ``Santander-UCIIIM Excellence Chairs 2011" for supporting his stay at the University Carlos III de Madrid during the development of this research.
MAR and PT wish to thank the support by the research project FIS2011--22566, Ministerio de Ciencia e Innovaci\'{o}n, Spain.}

%\date{April 28, 2014}

\maketitle

\begin{abstract}
We present an algebraic formulation of the notion of integrability of dynamical systems, based on a nilpotency property of its flow: it can be explicitly described as a polynomial on its evolution parameter.  Such a property is established in a purely geometric--algebraic language, in terms both of the algebra of all higher-order constants of the motion (named the nilpotent algebra of the dynamics), and of a maximal Abelian algebra of symmetries (called a Cartan subalgebra of the dynamics).  It is shown that this notion of integrability amounts to the annihilator of the nilpotent algebra being contained in a Cartan subalgebra of the dynamics. Systems exhibiting this property will be said to be nilpotent integrable.

Our notion of nilpotent integrability offers a new insight into the intrinsic dynamical properties of a system, which is independent of any auxiliary geometric structure defined on its phase space. At the same time, it extends in a natural way the classical concept of complete integrability for Hamiltonian systems.

An algebraic reduction procedure valid for nilpotent integrable systems, generalizing the well-known reduction procedures for symplectic and/or Poisson systems on appropriate quotient spaces, is also discussed.

In particular, it is shown that a large class of nilpotent integrable systems can be obtained by reduction of higher-order free systems.
The case of the third-order free system is analyzed and a nontrivial set of third-order Calogero-Moser-like nilpotent integrable equations is obtained.
\end{abstract}

\tableofcontents

%\tableofcontents

\section{Introduction}

In this paper, a new approach to the notion of integrability of dynamical systems that extends the standard approach \`a la Liouville, and also includes such obvious ``integrable'' systems like uniformly accelerated ones, will be discussed.   For reasons that will become evident later on, this new notion of integrability will be called \textit{nilpotent integrability}.

One of the main features of nilpotent integrability is that it does not require any auxiliary underlying geometrical structure for its formulation.  In this perspective,  nilpotent integrability becomes an intrinsic property of a dynamical system, and the specific geometry associated to it comes later.   Moreover, we shall show that nilpotent integrability is consistent with the different procedures of reduction available for dynamical systems. They will allow to create new nilpotent integrable systems out of simple examples.
%If we reduce a given dynamical system for instance by selecting an invariant submanifold and the system was nilpotent integrable, the reduced system will still be nilpotent integrable.
%In this sense the notion of nilpotent integrability is consistent with the natural operations exerted on dynamical systems.

The main motivation behind the idea that  `integrability', as a property of a dynamical system, must be prior to any geometry associated to it, comes from a most pragmatic approach to integrability.  For instance, if we demand that a system is `integrable' if its flow could be explicitly exhibited by using a finite number of quadratures, there is no obvious geometrical content in such a definition.

As a logical consequence of this observation, it is  natural to try to reformulate the well established notion of complete integrability of Hamiltonian systems without appealing to any auxiliary geometry, such as a symplectic or a Poisson structure (see, for instance, \cite{Mm85} and references therein).  The consistency of these ideas with general reduction procedures would reflect the fact that quadratures and elimination of variables should be  compatible with each other.

To clarify the scope of the ideas leading to the notion of nilpotent integrability, let us consider first the instance of a complete integrable Hamiltonian system system $\Gamma$ from the standard perspective. Thus, if $\Gamma$ represents a complete integrable Hamiltonian system on a symplectic manifold,  introducing action-angle variables $h_j$, $\theta^j$, $j = 1, \ldots, n$, the vector field describing $\Gamma$ will look like
$$
\Gamma = \nu^j (h) \frac{\partial}{\partial \theta^j} \,.
$$
The explicit expression of its flow $\varphi_t$, given by
$$
h_j \circ \varphi  = h_j \, , \qquad \theta^j \circ \varphi_t = \theta^j + \nu^j( h) t \, ,
$$
shows that the flow $\varphi_t$, described in action-angle variables, is affine in $t$.   More important, it is trivial to check that the derivatives along the dynamics of the chosen coordinates are given by
\begin{equation}\label{comp_int_variables}
\mathcal{L}_\Gamma \theta^j = \nu^j (h) \, ; \qquad  \mathcal{L}_\Gamma h_j = 0 \, .
\end{equation}
The previous equations can also be interpreted by saying that the representation in action-angle variables of the derivation $\Gamma$ is nilpotent of order 2.     It is relevant to point out that Eqns.\eqref{comp_int_variables} can be stated as a definition of complete integrability without any reference to a symplectic structure.

Let us consider now simple examples of systems which are `integrable' in an obvious sense but that do not fit naturally in the previous scheme.  For instance, let us consider a uniformly accelerated system in $\mathbb{R}^3$:  $ \dddot{\mathbf{r}} = 0$, $\mathbf{r} \in \mathbb{R}$.

The equations of motion of such a system are described by the vector field on $(\mathbf{r} ,\mathbf{v},\mathbf{a} ) \in T^2\mathbb{R}^3$
$$
\Gamma = \mathbf{v}\cdot \frac{\partial }{\partial \mathbf{r} } + \mathbf{a} \cdot  \frac{\partial }{\partial \mathbf{v} }  \, ,
$$
with $\mathbf{a}$ a constant vector.  Its flow $\varphi_t$ has the explicit polynomial expression
$$
\mathbf{r} \circ \varphi_t = \mathbf{r} + \mathbf{v} t + \frac{1}{2}\mathbf{a} t^2 \, .
$$
Notice that the acceleration $\mathbf{a}$ is obviously a constant of the motion, while $\mathbf{v}$ and $\mathbf{r}$ are higher-order constants of the motion, that is $\mathbf{v} = \dot{\mathbf{r}}$, $\mathbf{a} = \dot{\mathbf{v}} =\ddot{\mathbf{r}}$ and $\dot{\mathbf{a}} = 0$.  Similarly,
the kinetic energy $E_1 = \frac{1}{2} \mathbf{v}^2$ is not anymore a constant of the motion. However,
$$
E_2 = \dot{E}_1 = \mathbf{v} \cdot \mathbf{a} \, , \qquad E_3 = \dot{E}_2 = \mathbf{a}^2 \, , \qquad \dot{E}_3 = 0 \, ,
$$
i.e., $E_3$ is a constant of the motion, while $E_2$, $E_1$ are higher-order constants of the motion (of orders 2 and 3 respectively). The same can be argued with respect to the tower of (higher-order) constants of the motion derived from the standard angular momentum:  $\mathbf{L}_1 = \mathbf{r} \wedge \mathbf{v}$, $ \mathbf{L}_2 =  \dot{\mathbf{L}}_1 = \mathbf{r} \wedge \mathbf{a}$, $\mathbf{L}_3 =  \dot{\mathbf{L}}_2 = \mathbf{v} \wedge \mathbf{a}$, $\dot{\mathbf{L}}_3 = 0$.

The way how the previous ideas fit with the notion of reduction can be nicely described by considering the basic example of
the completely integrable free system $\ddot{\mathbf{r}} = 0$ in $\mathbb{R}^3$.   It is easy to construct other, perhaps more interesting, completely integrable systems by reducing it using its integrals of motion.  For instance, by using the standard polar decomposition of $\mathbf{r}$,
\begin{eqnarray*}
\mathbf{r} = r \mathbf{n} \, , \qquad && \mathbf{n} \cdot \mathbf{n} = 1 \, ; \\
\dot{\mathbf{r}} = \dot{r} \mathbf{n} + r \dot{\mathbf{n}} \, , \qquad && \mathbf{n} \cdot \dot{\mathbf{n}} = 0 \, ;\\
\ddot{\mathbf{r}} = \ddot{r} \mathbf{n} + 2\dot{r} \dot{\mathbf{n}}  + r \ddot{\mathbf{n}} \, , \qquad && \mathbf{n} \cdot \ddot{\mathbf{n}} = -\dot{\mathbf{n}}^2 \, ,
\end{eqnarray*}
we may write its constants of the motion $E_1$ and $\mathbf{L}_1$ as
$$
E_1 = \frac{1}{2} \left( \dot{r}^2 + r^2 \dot{\mathbf{n}}^2 \right) \, ,\qquad l_1^2 = \mathbf{L}_1^2 = r^4 \dot{\mathbf{n}}^2 \,.
$$
Then, once we restrict to the submanifold
$$
\Sigma (k,\alpha ) = \{ (\mathbf{r}, \mathbf{v}) \in T\mathbb{R}^3 \mid \alpha l_1^2 + 2(1-\alpha ) E_1 =  k, \, 0 \leq \alpha \leq 1  \} \, ,
$$
since for our system we have $\ddot{r} = r \dot{\mathbf{n}}^2$, we get the following family of completely (and explicitly) integrable reduced second-order systems in one-dimension:
$$
\ddot{r} = \frac{\alpha l_1^2 + 2(1- \alpha)E_1-(1- \alpha) \dot{r}^2}{r(\alpha r^2 + (1-\alpha))} \, .
$$
It is immediate to generalize the previous construction to the system $\dddot{\mathbf{r}} = 0$.  Now we have the additional relations
$$
\dddot{\mathbf{r}} = \dddot{r} \mathbf{n} + 3 \ddot{r} \dot{\mathbf{n}} + 3 \dot{r} \ddot{\mathbf{n}} + r \dddot{n} \, , \qquad \mathbf{n} \cdot \dddot{\mathbf{n}} = - 3 \dot{\mathbf{n}} \cdot \ddot{\mathbf{n}} \, .
$$
Moreover, we have $\dddot{r} = 3\dot{r} \dot{\mathbf{n}}^2 - r \mathbf{n} \cdot \dddot{\mathbf{n}} $.  The towers of higher-order constants of the motion $E_1, E_2, E_3$ and $\mathbf{L}_1,  \mathbf{L}_2, \mathbf{L}_3$ before defined, allow us to write e.g. $E_2 = \dot{\mathbf{n}}\cdot \ddot{\mathbf{n}} =  \dot{r} \ddot{r} + r \dot{r}\dot{\mathbf{n}}^2 + r^2 \dot{\mathbf{n}} \cdot \ddot{\mathbf{n}}$. Hence we get the explicitly integrable third order system:
$$
\dddot{r} = 3\frac{E_2 - \dot{r}\ddot{r}}{r} \, , \qquad E_2(t) = a^2t + c \, .
$$
In other words, the reduction of systems like $\dddot{\mathbf{r}} = 0$, possessing polynomial flows, gives rise to systems that can be explicitly integrated and possess polynomial flows as well.   Additional examples of this kind will be discussed later on.

Therefore, the main idea behind the notion of nilpotent integrability is to take advantage not only of the constants of the motion, but also of the whole family of higher-order integrals of motion associated with the system.  If this family (that will be called the nilpotent algebra of the dynamical system) is large enough, then the flow of the system will be shown to be explicitly described, in terms of the integrals, as a polynomial function of time.  In such a case, the system will be said to be \textit{strictly nilpotent integrable}.   Systems obtained by reduction of the system $\dddot{\mathbf{r}} = 0$  will provide examples of the strictly nilpotent case.

However, it could also happen that the nilpotent algebra associated with a given system would not be large enough, although there exists  an Abelian algebra of symmetries of the dynamics large enough to complement it.  For instance, in the case of completely integrable systems such Abelian algebra of symmetries is the algebra of Hamiltonian vector fields defined by the action variables, whose flow parameters are the ``angle'' variables $\theta^j$.   Such systems will be called the completely nilpotent integrable ones. A particular instance of them are the standard completely integrable Hamiltonian systems to which we have referred to.

Due to the paramount role played in this analysis by the algebraic structure of the families of higher-order constants of the motion and various Lie algebras of symmetries of the system, we are led to develop the theory in an algebraic framework, finding also inspiration in the notion of \textit{algebraic integrability} introduced in \cite{La93}.

It is important to mention that the algebraic framework also allows to extend  in a meaningful way the notion of  `integrability'' to quantum systems (see for instance \cite{Cl09}, \cite{Fa10}, and the discussion in \cite{Ca14}, Chaps. 7 and 8).  In fact, the algebraic approach to the description of physical systems has gained a significant weight since the discovery of  quantum group symmetries in physical systems.   More recent developments in quantum gravity (see for instance \cite{La97}, \cite{Va06}, \cite{Ba10}) point out the need for a non--commutative description of space-time, hence emphasizing the algebraic approach to their study.  In all these contexts, the algebraic view-point instead of a set--theoretical description of physical systems is mandatory.

Hence, it is all but natural to pursue a foundation of the theory of dynamical systems from this point of view.
Even more, in doing so, as the following discussion will show, new insight and ideas can emerge naturally.

Consequently, we shall introduce the notions of strict and complete \textit{nilpotent integrability} of a dynamical system $\Gamma$ by identifying it with a derivation of an algebra (of an appropriate class) $\mathcal{F}$.  We will prove that, under adequate conditions, the flow of such system can be represented as a polynomial function of time.   As was mentioned above, our theory applies to a large class of dynamical systems, possessing no obvious geometrical structure, like a symplectic, or Poisson one, etc.   The algebraic framework will also make transparent the consistency of generalized reduction procedures (i.e. reduction with respect to subalgebras and ideals compatible with the dynamics) with the notion of nilpotent integrability.  Hence we shall prove that, under appropriate hypotheses, if a system is nilpotent integrable its reductions are nilpotent integrable too.

The formalism of \textit{differentiable spaces} is the language we have chosen to deal with the algebraic formulation of dynamics.   Differentiable spaces are one of the most useful extensions of the notion of differentiable manifolds that permit a simple algebraic description.  In fact, local models for differentiable spaces are just differentiable algebras, i.e., quotients of algebras of smooth functions on open sets in finite dimensional Euclidean spaces equipped with the strong Whitney topology.    The use of differentiable spaces will also allow to keep the formalism close to the standard description of dynamical systems on manifolds.   Another advantage of this approach is that a large class of mild singularities are automatically dealt with within this formalism. Also, some technical difficulties in dealing with local charts in the standard formalism of differentiable manifolds are overlooked.  It is also useful to mention here that the present algebraic approach in principle can be extended to the case of dynamical systems in the context of Galois differential algebras along the lines discussed for instance in \cite{Te13}.

The structure of the paper is the following. The first subsections of Section \ref{sec:algebras_etc} will be devoted to review the standard notions of dynamical systems both in the language of smooth manifolds and of differentiable spaces, making the transition as smooth as possible.    In Section \ref{sec:symmetries_etc}, we shall discuss the fundamental structures attached to a dynamical system, i.e., its algebra of symmetries and its nilpotent algebra, and we will introduce the fundamental notion of nilpotent integrability.  In Section \ref{sec:reduction_etc}, a theory of generalized reductions will be formulated in the framework of nilpotent integrability. Finally, in Section \ref{sec:examples} a family of examples leading to a Calogero-Moser higher-order nilpotent integrable system will be discussed.

%%%%%%%%%%%%%%%%%%%%%%%%%%
%%%%%%%%%%%%%%%%%%%%%%%%%%

\section{Algebras, derivations and dynamical systems}\label{sec:algebras_etc}

\subsection{Dynamical systems in smooth manifolds}  Let $M$ be a smooth paracompact differentiable manifold of dimension $n$.   Local coordinates on $M$ will be denoted by $x^i$, $i = 1,\ldots,n$.  Let $\mathcal{F}(M)$ be the commutative algebra of smooth functions on $M$. Hereafter we will restrict to real valued functions (although in some applications complex valued functions may also be relevant).

Dynamics on $M$ will be described by a vector field $\Gamma$, which is a derivation of the algebra $\mathcal{F}(M)$, i.e. a linear map on $\mathcal{F}(M)$ such that $\Gamma (fg) = \Gamma(f) g + f \Gamma (g)$.  The space of derivations of the algebra $\mathcal{F(M)}$, denoted in what follows by $\mathrm{Der}(\mathcal{F(M)})$, is a module over the algebra $\mathcal{F(M)}$. Moreover, it converts into a Lie algebra if equipped with the Lie bracket
\[
[Y_1,Y_2](f) = Y_1(Y_2(f)) - Y_2(Y_1(f)), \quad \forall \hspace{1mm}  Y_1,Y_2 \in \mathrm{Der}(\mathcal{F(M)}), \quad f \in \mathcal{F(M)}.
\]

The dynamics represented by $\Gamma$ will be described in two possible ways.

\noi a) As a system of autonomous first order differential equations, that in local coordinates $x^i$ takes the form:
\begin{equation}\label{dynamics}
 \frac{dx^i}{dt} = \Gamma^i(x) ,\quad i = 1, \ldots, m ,
 \end{equation}
where the vector field locally reads $\Gamma = \Gamma^i(x) \partial / \partial x^i$.

\noi b) Equivalently, as an evolution equation on $\mathcal{F}(M)$:

\begin{equation}\label{evolution1}
 \frac{df}{dt} = \Gamma (f) , \quad   f \in \mathcal{F}(M) .
\end{equation}

\noindent The local flow of $\Gamma$ will be denoted by $\varphi_t$ and satisfies:
\begin{equation}\label{local_flow1}
 \frac{d\varphi_t}{dt} = \Gamma \circ \varphi_t .
 \end{equation}
If the vector field is complete, the flow $\varphi_t$ defines a one--parameter group of automorphisms of the algebra $\mathcal{F}(M)$ (a one--parameter family of automorphisms for non--autonomous systems).
Alternatively, we can integrate formally eq. (\ref{evolution1}) to obtain:
\begin{equation}\label{formal_integration1}
\varphi_t^* f = f \circ \varphi_t = \exp (t\Gamma) f = \sum_{n = 0}^\infty  \frac{t^n}{n!}  \Gamma^n(f).
\end{equation}

\subsection{Algebras of functions and differentiable algebras}
The algebra $\mathcal{F}(M)$ of smooth functions on the manifold $M$ is equipped with the topology of uniform convergence on compact sets of the function and their derivatives to order $r$ $(r = 1,2,\ldots, \infty )$. However, for the purpose of extending the standard notions of differential calculus, it is necessary to consider the strong or Whitney topology on $\mathcal{F}(M)$, that is the topology defined by the countable family of seminorms $p_{ik}(f) = || \partial^{|k|} f / \partial x^k ||_{L^\infty (K_i)}$, where $\{  K_i \}$ is a subcovering by compact sets of a small enough open covering $\{ U_i \}$ of $M$. Thus $\mathcal{F}(M)$ becomes a Fr\'echet topological algebra, i.e., a locally multiplicative convex algebra which is metrizable and complete, where multiplicative convex means that the topology is defined by a countable family of submultiplicative seminorms.

The algebra $C^{\infty}(\mathbb{R}^n)$ is an instance of a Fr\'echet topological algebra and the model for the notion of a differentiable algebra \cite{Ma66} (see also \cite{Na03} for a recent account of the theory of differentiable algebras).

%The algebra $\mathcal{F}(M)$ of smooth functions on the manifold $M$ is equipped with the topology of uniform convergence on compact sets of functions and their derivatives to order %$r$ $(r = 1,2,\ldots, \infty )$ and it acquires the structure of a Fr\'echet topological algebra.
%
%Moreover, we will endow $\mathcal{F}(M)$ with the Withney topology, defined by the countable family of seminorms $p_{ik}(f) = || \partial^{|k|} f / \partial x^k ||_{L^\infty (K_i)}$, %where $\{  K_i \}$ is a subcovering by compact sets of an open covering $\{ÊU_i \}$ of $M$.
%
%The algebra $\mathcal{F}$ is an instance of a differentiable algebra \cite{Ma66}.

\begin{definition}
A differentiable algebra is a real Fr\'echet algebra isomorphic to $C^\infty( \mathbb{R}^m)/\mathcal{J}$, where $\mathcal{J}$ is a closed ideal of $C^\infty(\mathbb{R}^m)$.
\end{definition}
  Notice that the class of differentiable algebras is larger than the class of algebras of differentiable functions on a manifold (for instance $C^\infty (\mathbb{R})/(x^2)$ is not an algebra of differentiable functions on a smooth manifold).

 Because of Whitney's embedding theorem, any smooth paracompact manifold $M$ can be embedded into $\mathbb{R}^m$, hence $\mathcal{F}(M) \cong C^\infty(\mathbb{R}^m)/\mathcal{J}$, where $\mathcal{J}$ is the closed ideal of smooth function on $\mathbb{R}^m$ vanishing on $M$ (as a closed submanifold of it).   Hence $\mathcal{F} = C^\infty (M)$ is a differentiable algebra too.
\begin{definition}
Given a real topological algebra $\mathcal{F}$, we define its spectrum $\spec_\R (\mathcal{F})$ as the space of continuous $\R$--morphisms $\varphi \colon \mathcal{F} \to \R$ equipped with the natural topology.
\end{definition}
This space can be identified with the space of real maximal ideals of $\mathcal{F}$.

 When $M$ is a Hausdorff smooth manifold satisfying the second countability axiom, then $\spec_\R (\mathcal{F}(M)) \cong M$.   Given an element $f \in \mathcal{F}$ and a point $x \in \spec_\R\mathcal{F}$ we define the Gel'fand transform $\hat{f}(x) = x(f)$. Thus, any $f$ defines a function on $\spec_\R \mathcal{F}$. If we consider $\mathcal{F}(M)$, then $\hat{f}$ coincides with $f$.

The spectrum of a differentiable algebra is the local model for a class of spaces called $\mathcal{C}^\infty$--spaces or differentiable spaces for short, that extend in a natural way the notion of smooth manifolds (see for instance \cite{Na03} and references therein for a modern account of the foundations of the theory of differentiable spaces).

Differentiable spaces are ringed spaces over a sheaf of differentiable algebras. They allow to work with a large class of topological spaces, which share with smooth manifolds most of their theoretical formulation and the differential calculus attached to them. At the same time, differentiable spaces can exhibit singularities and other structures typically arising in the standard processes of reduction of dynamical systems. As was commented in the introduction, this is the reason  why we have chosen such a structure to introduce the notion of nilpotent integrability.   Actually, we will use in this paper the ``local'' formulation of the theory only, i.e., the notion of differentiable algebras.  However, once formulated in the case of differentiable algebras, nilpotent integrability can be extended naturally to the category of differentiable spaces.

Differentiable spaces defined as the real spectrum of a differentiable algebra are called \textit{affine differentiable spaces}. It can be proved that a differentiable space is affine if and only if it is separated. This means that at each point there is a finite local family of functions separating infinitely near points (see later, Def. 2), with bounded dimension, and its topology has a countable basis (see the ``Embedding Theorem'' at \cite{Na03}, p. 67).  Thus, any compact separated differentiable space is affine and a standard smooth manifold can be characterized as a separated, finite-dimensional differentiable space with a countable basis for its topology.

Given a differentiable algebra $\mathcal{F}$ and a point $x \in \spec_\R \mathcal{F}$, let $\mathfrak{m}_x$ be the maximal ideal  of all $f\in \mathcal{F}$ such that $\hat{f}(x) = 0$.    Also, let $\mathcal{F}_x := \mathcal{F}_{\mathfrak{m}_x}$ be the localization of $\mathcal{F}$ with respect to the multiplicative system $S_x$ defined by $\mathfrak{m}_x$, namely the set of all elements $f \notin \mathfrak{m}_x$.
\begin{definition}
We shall call the elements on $\mathcal{F}_x$ the \textit{germs} of elements $f$ of $\mathcal{F}$ at $x$.
\end{definition}
They can be considered as residue classes of elements on $\mathcal{F}$ with respect to the ideal of elements of $\mathcal{F}$ vanishing on open neighborhoods of $x$; we shall denote them by $[f]_x$.

For a differentiable algebra $\mathcal{F}$, the space of its derivations $\mathrm{Der}(\mathcal{F})$ can be equipped with its canonical Lie algebra structure $[\cdot , \cdot ]$.    Such a space is a $\mathcal{F}$--module. The space $\mathrm{Der}(\mathcal{F})^*$ of $\mathcal{F}$--valued $\mathcal{F}$--linear maps on $\mathrm{Der}(\mathcal{F})$ is the space of 1--forms of the differentiable algebra $\mathcal{F}$.   The space  $\mathrm{Der}(\mathcal{F})^*$ is again a $\mathcal{F}$--module; we will denote it by $\Omega^1(\mathcal{F})$ (or simply by $\Omega^1$ for short).

Let $U$ be an open set in a differentiable space $X$. All objects considered so far can be localized to $U$ in a standard way. Therefore, we can consider the differentiable algebras $\mathcal{F}_U$,  $\mathrm{Der}(\mathcal{F})_{U}$,  $\Omega^1_U$, etc.    Moreover, differentiable spaces admit partitions of unity (see \cite{Na03} p. 52).

The differential map $d \colon \mathcal{F} \to \Omega^1$ is defined as $df (Y)  = Y (f)$ for any $Y\in \mathrm{Der}(\mathcal{F})$.
In the case of $\mathcal{F} = C^\infty (M)$, the space $\Omega^1$ is just the module of smooth 1--forms $\Omega^1(M)$ on the manifold $M$.

Alternatively, we can associate with any point $x \in \spec_\R \mathcal{F}$ the linear space $\mathfrak{m}_x / \mathfrak{m}_x^2$ denoted by $T_x^*X$ where $X = \spec_\R\mathcal{F}$.    Again, the space $T_x^*X$ is the localization of the module $\Omega^1$ with respect to the multiplicative system defined by $\mathfrak{m}_x$.     Then given $f \in \mathfrak{m}_x$, we have that $df(x) := [df]_x$ can be identified with the class of $f$ in $\mathfrak{m}_x/\mathfrak{m}_x^2$.  In the same vein it is possible to show that the dual space to $T_x^*X$, namely  the tangent space $T_xX$ to the differentiable space $X$  at $x$,  is just the space of derivations of the localized algebra $\mathcal{F}_x$.

%%%%%%%%%%%%%%%%%%%%%%%%
%%%%%%%%%%%%%%%%%%%%%%%%

\subsection{Generating sets}
\begin{definition}
Given a differentiable algebra $\mathcal{F}$, we shall say that a family of functions $g_1, \ldots, g_n$ separate points infinitely near  to $x \in X = \spec_\R\mathcal{F}$ if the differentials ${dg_1}(x), \ldots, dg_n(x)$ generate $T_x^*X$ as a linear space.
\end{definition}

\noi Equivalently, the differentials $dg_k(x)$ of the family separate elements  $v_x$ in $T_xX$, i.e., for any $v_x \neq w_x$ there is $g_k$ such that $v_x(g_k) \neq w_x(g_k)$.
\begin{definition}
We say that the family of functions $\mathcal{G} = \{ g_\alpha \}$ of the differentiable algebra $\mathcal{F}$ is a differential generating set if they separate points in $X = \spec_\R \mathcal{F}$ and if for any $x\in X$, there is a finite subfamily $\{g_{\alpha_1}, \ldots, g_{\alpha_n}\}$ of $\mathcal{G}$ that separate infinitely near points to $x$.
\end{definition}
Notice that if $\mathcal{G}$ is a differential generating set for the differentiable algebra $\mathcal{F}$,
then for each pair of derivations $Y_1$ and $Y_2$ there exists at least one element $g_\alpha \in \mathcal{G}$ such that $Y_1 (g_\alpha ) \neq Y_2 (g_\alpha)$.

When the previous condition holds, we say that the functions $g_\alpha$ \textit{separate derivations}.
It is not hard to see that a large class of manifolds, for instance $C^\infty$ paracompact second countable manifolds, possesses generating sets.
Also, if  $\mathcal{G}$ is a generating set for the algebra $\mathcal{F}$ as an associative algebra, then it is a differential generating set for $\mathcal{F}$.

Since derivations can be localized, generating sets also separate local derivations.
This fact is used to prove the following important property of differential generating sets, which justifies their name.

\begin{lemma}\label{lemma1}  Let $M$ be a smooth manifold. A set of functions $\mathcal{G} = \{ g_\alpha \}$ of $\mathcal{F}(M)$ is a differential generating set iff
the set of 1--forms $\mathrm{d}\mathcal{G} = \{ \mathrm{d}g_\alpha \}$ generates the algebraic dual of the space of derivations $\mathrm{Der}(\mathcal{F})^*$, which coincides with the space $\Omega^1(M)$ of smooth 1--forms on $M$.
\end{lemma}
\begin{proof}
Consider first an open set $U$ contained in the domain of a local chart of $M$. If there were a local 1--form $\sigma_U$ that could not be written as $\sigma_U = \sigma_{U,\alpha} dg_\alpha$, then the span of $dg_\alpha$ would be a proper subspace of $T^*U$ and it would exist a vector field $X$ lying in the annihilator of such subspace. Hence $dg_\alpha (X) = 0$ for all $\alpha$.    This argument can be made global by using partitions of the unity: for any 1--form $\sigma$ on $M$ it must exist a family of functions $\sigma_\alpha$ with compact support on $M$ such that $\sigma = \sigma_\alpha \mathrm{d} g_\alpha$.
\end{proof}

In the subsequent analysis, we shall assume the existence of a differential generating set $\mathcal{G}$ for $\mathcal{F}$.
Let us mention another direct consequence of the properties of a generating set $\mathcal{G}$. Since locally the differentials $dg_\alpha$ generate $T^*X$ with $X$ the spectrum of $\mathcal{F}$, we can extract a subset $dg_{\alpha_i}$ that is locally independent, i.e. such that $dg_{\alpha_1}\wedge \cdots \wedge dg_{\alpha_m} \neq 0$ on a open neighborhood of any given point. Therefore, we will say that the functions $g_{\alpha_k}$ define
a local coordinate system.   Later on we will use this fact to write explicit formulas in terms of subsets of generating sets.

We will introduce a notion of finiteness for algebras $\mathcal{F}$ which is fundamental in order to define our concept of integrability.

%It also follows easily from the existence of generating sets for $\mathcal{F}$ that generically the points at which all functions of a generating %set coincide are isolated
\begin{definition}
The differentiable algebra $\mathcal{F}$ is said to be of finite type if it admits a finite differential generating set $\mathcal{G} = \{ g_1, \ldots, g_N \}, N\in\mathbb{N}$.
\end{definition}

This more restrictive condition is satisfied for instance if the manifold is compact or of finite type. Indeed, in such case it can be embedded into a finite dimensional Euclidean space whose coordinate functions, restricted to the embedded manifold, provide a differential generating set.
As a consequence of the previous discussion, it can be shown that if $M$ is a smooth manifold, a differential generating set for the differentiable algebra $\mathcal{F}(M)$ provides a set of local coordinate systems, i.e., an atlas for the manifold $M$, by restricting to small enough open sets and shieving out dependent functions.

\subsection{Derivations and their flows on differentiable algebras}
 Let $\Gamma$ be a derivation of the differentiable algebra of finite type $\mathcal{F}$, hence $\mathcal{F} \cong \mathcal{C}^\infty (\R^n)/ \mathcal{J}$.    Denoting by $X$, as before, the real spectrum of $\mathcal{F}$, for each $x \in X$ the derivation $\Gamma$ defines an element $\Gamma_x \in T_xX$.  Hence, $X$ is a closed differentiable subspace of $\R^n$, and the canonical injection $i\colon X \to \R^n$ maps $\Gamma_x$ to a tangent vector $i_*\Gamma_x \in T_{i(x)} \R^n$. Moreover, we can extend the vector field $i_*\Gamma$ along $i(X)$ to a vector field $\widetilde\Gamma$ in $\R^n$.   Let $\tilde{\varphi}_t$ be the flow of $\widetilde{\Gamma}$. By construction $\tilde{\varphi}_t$ leaves $i(X)$ invariant.  Let us denote by $\varphi_t$ the restriction of $\tilde{\varphi}_t$ to $X$ (that always exists because of the universal property of closed differentiable subspaces \cite{Na03} p. 60).  We will denote by $\varphi_t$ the flow of the derivation $\Gamma$.   The flow $\varphi_t$ will act on elements $f\in\mathcal{F}$ as
$$ \varphi_t^*(f) = \tilde{\varphi}_t^*(\tilde{f}) + \mathcal{J}, \quad  \tilde{f} + \mathcal{J} = f.$$
Also, the flow $\varphi_t$ can be integrated formally by using a close analog to formula (\ref{formal_integration1}):
\begin{equation}\label{formal_integration}
\varphi_t^*(f) = \sum_{n = 0}^\infty  \frac{t^n}{n!}  \widetilde{\Gamma}^n(\tilde{f}) + \mathcal{J}.
\end{equation}
We have proved that not only  the derivations on differentiable algebras can be extended to derivations on $\mathcal{C}^\infty (\R^n)$, but also that they are continuous maps with respect to the canonical topology on $\mathcal{F}$. Indeed, they are the quotient of derivations on $\mathcal{C}^\infty (\R^n)$ preserving the closed ideal $\mathcal{J}$, and are continuous.

Finally, we observe that formulas similar to Eqs. (\ref{evolution1})-(\ref{local_flow1}) hold for derivations on differentiable algebras:

\begin{equation}\label{evolution}
 \frac{df}{dt} = \Gamma (f) , \quad   f \in \mathcal{F},
\end{equation}
with the local flow $\varphi_t$ satisfying
\begin{equation}\label{local_flow}
 \frac{d\varphi_t}{dt} = \Gamma \circ \varphi_t .
 \end{equation}

\section{The algebraic formulation of nilpotent integrability}\label{sec:symmetries_etc}

\subsection{Infinitesimal symmetries}
\begin{definition}  Let $\mathcal{L}$ be a Lie algebra.   Given a subset $\mathcal{S}\subset \mathcal{L}$, the space of infinitesimal symmetries of $\mathcal{S}$ is its commutant $\mathcal{S}'$, i.e.
$$
\mathcal{S}' = \{ \xi \in \mathcal{L} \mid [\xi , x] = 0, \forall x \in \mathcal{S} \} \, .
$$
In particular, if $\mathcal{L} = \mathrm{Der}(\mathcal{F})$ is the Lie algebra of derivations of a differentiable algebra $\mathcal{F}$, and $\mathcal{S} = \{ \Gamma \}$, with $\Gamma \in \mathrm{Der}(\mathcal{F})$, the space of infinitesimal symmetries of $\Gamma$ is the commutant of $\mathcal{S} = \{ \Gamma \}$, i.e., the set of derivations $\xi$ of $\mathcal{F}$ such that $[\Gamma, \xi ] = 0$.
\end{definition}

The space of infinitesimal symmetries of any subset $\mathcal{S}\subset \mathcal{L}$ is a Lie subalgebra of $\mathcal{L}$.

We define recursively the $n$th commutant of $\mathcal{S}$  as:
\begin{equation}
\mathcal{S}' = \{ \xi \in \mathcal{L} \mid [\xi, x ] = 0, \mathrm{\forall} \, \,  x \in \mathcal{S} \}, \quad \mathcal{S}^{(k + 1)} = (\mathcal{S}^{(k)})', \quad k \geq 1 .
\end{equation}
Notice that by definition $\mathcal{S} \subset \mathcal{S}''$ where $\mathcal{S}''$ is called the bicommutant of $\mathcal{S}$.    The following result holds.

\begin{lemma}[Stability Lemma]\label{stability1}
Let $\mathcal{L}$ be a Lie algebra and $\mathcal{S} \subset \mathcal{L}$ a subset.  Then:
\begin{equation}
\mathcal{S}\label{tower} \subset \mathcal{S}'' = \mathcal{S}^{(4)} = \mathcal{S}^{(6)} =  \cdots, \quad   \mathcal{S}' = \mathcal{S}''' = \mathcal{S}^{(5)} =  \cdots .
\end{equation}
Moreover, if $\mathcal{S}$ is Abelian, we have $\mathcal{S} \subset \mathcal{S}''  \subset \mathcal{S}'$ and the bicommutant $\mathcal{S}''$ is an Abelian Lie subalgebra of $\mathcal{L}$.
\end{lemma}

\begin{proof}
Observe that if $\mathcal{S}_1 \subset \mathcal{S}_2$, then $\mathcal{S}_2' \subset \mathcal{S}_1'$.  Because $\mathcal{S} \subset \mathcal{S}''$, we have $\mathcal{S}''' \subset \mathcal{S}'$.  On the other hand, since for any set $\mathcal{S} \subset \mathcal{S}''$ we have that $\mathcal{S}' \subset (\mathcal{S}')''$, we infer that $\mathcal{S}' = \mathcal{S}'''$ and relations (\ref{tower}) follow.

If $\mathcal{S}$ is Abelian, i.e., $[x,y] = 0$ for all $x,y \in \mathcal{S}$, then $\mathcal{S} \subset \mathcal{S}'$, hence $\mathcal{S}'' \subset \mathcal{S}'$.  Moreover if $\xi, \zeta \in \mathcal{S}'' \subset \mathcal{S}'$, then $\zeta \in \mathcal{S}'$, and $[\xi, \zeta] = 0$, which shows that $\mathcal{S}''$ is Abelian.
\end{proof}
\begin{definition}
Let $\Gamma$ be a derivation of the differentiable algebra $\mathcal{F}$.  We shall say that $\Gamma$ is quasi-periodic if the commutant $\mathcal{S}'_\Gamma$ is finite dimensional.
\end{definition}
Notice that if $\Gamma$ is a vector field on a compact manifold $M$, then $\Gamma$ is quasi-periodic if the flow of $\Gamma$ generates a finite dimensional torus inside the group of diffeomorphisms of $M$ whose Lie algebra is the commutant $\mathcal{S}'_\Gamma$.

However, we will be mainly interested in the ``Abelian part'' of the algebra of symmetries of $\Gamma$, that is we will consider a \textit{maximal Abelian subalgebra}  of $\mathcal{S}'$ containing $\mathcal{S}$.  When the subset $\mathcal{S}$ is just $\Gamma$, we shall denote such a maximal Abelian subalgebra by $\mathfrak{h}_\Gamma$.  Hence, $\mathfrak{h}_\Gamma$ is a Cartan subalgebra of $\mathcal{S}'$ containing $\Gamma$ only; consequently, it will be called a Cartan subalgebra associated to the dynamics $\Gamma$.

More can be said when the Lie algebra $\mathcal{L}$ is represented in terms of an algebra of derivations of a given algebra $\mathcal{F}$ (or of automorphisms of an $\mathcal{F}$-module).   Let $\rho\colon \mathcal{L} \to \mathrm{Der}(\mathcal{F})$ be a morphism of Lie algebras.  To each element $\xi \in \mathcal{L}$ we associate a derivation $\rho (\xi)$, satisfying $\rho([\xi, \zeta]) = [\rho (\xi), \rho (\zeta)]$.  We will analyse this situation in the following sections.

\subsection{Constants of the motion}

\begin{definition}
Given the dynamics $\Gamma$ on $\mathcal{F}$, the subalgebra of constants of the motion of $\Gamma$, denoted by $\mathcal{C}(\Gamma)$ (or simply by $\mathcal{C}$) is defined by
\begin{equation}\label{constants_motion}
\mathcal{C}(\Gamma) = \{  f\in \mathcal{F}  \mid \Gamma (f) = 0  \}.
\end{equation}
\end{definition}

More generally,  as in the previous section we can consider a Lie algebra $\mathcal{L}$ represented by derivations of the algebra $\mathcal{F}$ (in particular we can consider the Lie algebra $\mathrm{Der}(\mathcal{F})$ itself with the tautological representation). Given any subset $\mathcal{S} \in \mathcal{L}$, we define the subalgebra of its constants of the motion $\mathcal{C}(\mathcal{S})$ by
\begin{equation}
\mathcal{C}(\mathcal{S}) = \{ f \in \mathcal{F} \mid \rho(x) (f) = 0, \mathrm{\forall} \, \, x \in \mathcal{S} \}.
\end{equation}

\noi We will use a notation reminiscent of the one used in the previous section by denoting the subalgebra of constants of the motion by $\mathcal{S}_{\rho}' \subset \mathcal{F}$ and calling it the ``commutant'' of $\mathcal{S}$ in $\mathcal{F}$ with respect to the representation $\rho$.  Similarly, we can define the bicommutant of $\mathcal{S}$ in $\mathcal{F}$ as the Lie subalgebra of $\mathcal{L}$, denoted as $\mathcal{S}_{\rho}''$, of elements $\xi \in \mathcal{L}$ such that $\rho (\xi) (f) = 0$ for all $f \in \mathcal{C}(\mathcal{S}) = \mathcal{S}_{\rho}' \subset \mathcal{F}$.  We define recursively $\mathcal{S}_\mathcal{F}^{(p)} \subset \mathcal{F}$ as follows:
\begin{eqnarray}\label{Srho_commutant}
\mathcal{S}_{\rho}^{(1)} &=&  \mathcal{S}_{\rho}'  = \mathcal{C}(\mathcal{S}) \, , \qquad \mathcal{S}_{\rho}^{(2)} =  \mathcal{S}_{\rho}'' \, , \\ \nonumber
\mathcal{S}_{\rho}^{(2k+1)} & =& \{ f \in \mathcal{F} \mid  \rho(\xi ) (f) = 0, \forall \xi \in \mathcal{S}_{\rho}^{(2k)} \} \, , k \geq 1 \\  \nonumber \mathcal{S}_{\rho}^{(2k)} &=& \{ x \in \mathcal{L} \mid  \rho(\xi)(f) = 0, \forall f \in \mathcal{S}_{\rho}^{(2k-1)} \, ,  k \geq 2 \, .\}
\end{eqnarray}

\noi Then, an argument similar to that of  Lemma \ref{stability1} gives us another stability result.

\begin{lemma}\label{stability2}
The following relations hold
\begin{equation}
\mathcal{S} \subset \mathcal{S}_{\rho}'' = \mathcal{S}_{\rho}^{(4)} = \cdots , \quad \mathcal{C}(\mathcal{S})  =  \mathcal{S}_{\rho}'  =  \mathcal{S}_{\rho}''' = \cdots .
\end{equation}
Moreover, if $\mathcal{S}$ is an Abelian subset of $\mathcal{L}$, then $ \mathcal{S}_{\rho}'' $ is the minimal Abelian subalgebra of $\mathcal{L}$ which contains $\mathcal{S}$.
\end{lemma}

Observe that the kernel of any derivation of a given algebra is a subalgebra of it. Since $\mathcal{C}$ is a subalgebra of $\mathcal{F}$, if  $\mathcal{M}$ is a $\mathcal{F}$--module,  $\mathcal{M}$ is a $\mathcal{C}$--module too.
In particular, $\mathcal{F}$ is a $\mathcal{C}$--module.     The space of $\mathcal{C}$--linear maps from $\mathcal{F}$ to $\mathcal{C}$ is called the space of sections of $\mathcal{F}$ over $\mathcal{C}$ (see later for the set theoretical interpretation of such set).

If $\mathcal{F}$ is the differentiable algebra of smooth functions on a manifold $M$, then the algebra of constants of motion for the derivation $\Gamma$ is closed in $\mathcal{F}$ with respect to its canonical topology, because $\Gamma$ is a continuous map.  However, it is not true in general that a closed subalgebra of a differentiable algebra is a differentiable algebra.
In what follows, we shall assume  that the algebra of constants of motion for $\Gamma$ is a differentiable algebra.   This assumption is actually satisfied by a large family of derivations, as we will see in several examples.

Let $C = \spec_\R(\mathcal{C})$ denote the real spectrum of the differentiable algebra $\mathcal{C}$. Hence the canonical inclusion morphism $i\colon \mathcal{C} \to \mathcal{F}$ induces a continuous map $\pi \colon X \to C$, such that $i = \pi^*$.  The injectivity of $i$ implies that the map $\pi$ is a surjective projection map.       Let $c \in C$ be a point in the real spectrum of $\mathcal{C}$. Then $\mathcal{F}/ \mathcal{F} \mathfrak{m}_c$ is the differentiable algebra of the fibre $\pi^{-1}(c)$ of $\pi$ over the point $c$.   Notice that the flow of $\Gamma$ leaves invariant the fibres of $\pi$ and projects to the null derivation on $C$.

\subsection{The nilpotent algebra of a derivation}

We are able now to define the sets of higher-order constants of the motion or generalized constants of the motion.
\begin{definition}
We shall say that $f\in\mathcal{F}$ is a constant of the motion of order $k$ if
\begin{equation}
\lie{\Gamma}^k(f) = 0 \qquad and \qquad \lie{\Gamma}^{k-1} f \neq 0.
\end{equation}
The number $k$ will be called the nilpotency index of $f$ with respect to $\Gamma$.  The set of higher order constants of the motion of $\Gamma$ up to the index $k$ will be denoted by
$$ \mathcal{C}^k(\Gamma) = \{ f \in \mathcal{F} \mid \lie{\Gamma}^k (f) = 0  \}. $$
\end{definition}
With this notation, $\mathcal{C}(\Gamma) \equiv \mathcal{C}^1(\Gamma)$. Observe that the sets $\mathcal{C}^k(\Gamma )$  are closed linear spaces such that $\mathcal{C}^k (\Gamma ) \subset \mathcal{C}^l (\Gamma )$
for all $k < l$.  Consequently, we can construct the following filtration:

\begin{equation} \mathcal{C} (\Gamma ) =  \mathcal{C}^1(\Gamma ) \subset \cdots \subset \mathcal{C}^k (\Gamma ) \subset \cdots \subset \mathcal{F}.\label{filtr}
\end{equation}
\begin{lemma}
The spaces $\mathcal{C}^k(\Gamma )$ are $\mathcal{C}$--modules.
\end{lemma}
\begin{proof}
By using the binomial expansion of the derivation $\lie{\Gamma}^k$ acting on the product of two functions, we have
\begin{equation}
\lie{\Gamma}^k (fg) = \sum_{l = 0}^k \binom{k}{l}\lie{\Gamma}^l(f) \lie{\Gamma}^{k-l}(g) = 0 \label{binomial}
\end{equation}
for any $f \in \mathcal{C}$ and $g \in \mathcal{C}^k$.
It is also clear that $\lie{\Gamma}^l \mathcal{C}^k \subset \mathcal{C}^{k-l}$ for all $k > l$ and $\lie{\Gamma}^l \mathcal{C}^k = 0$ for all $l \geq k$.
\end{proof}
A central definition for the present theory is the following.
\begin{definition}
Given a dynamical system $\Gamma$ in the differentiable algebra $\mathcal{F}$, the nilpotent subspace of $\Gamma$ is the closure of the union of all generalized constants of the motion.  We will denote such set as $\mathcal{N}(\Gamma )$:
$$ \mathcal{N}(\Gamma) = \overline{\{ f \in \mathcal{F} \mid  \exists k \in \mathbb{N}, \lie{\Gamma}^k f  = 0  \} }.$$
\end{definition}

\noi The spaces $\mathcal{C}^k$ in general are not subalgebras, contrarily to the case of $\mathcal{N}$.
\begin{lemma}
The space $\mathcal{N}$ is a closed subalgebra of $\mathcal{F}$.
\end{lemma}
\begin{proof}
Let $f,g \in \mathcal{N}$, with nilpotency indices $k_1$ and $k_2$ respectively. Then by taking $p>k_1+k_2$, all terms in the binomial expansion (\ref{binomial}) of $\lie{\Gamma}^p (fg)$ have the form $\lie{\Gamma}^r (f) \lie{\Gamma}^s (g)$, with either $r$ or $s$ $>\max\{k_1,k_2\}$. Hence, all terms in the expansion must vanish.
\end{proof}
\noi Thus we have the following structure:
$$
\mathbb{R} \subset \mathcal{C} \subset \mathcal{N} \subset \mathcal{F}  \, .
$$

As in the case of the algebra of constants of the motion, we shall assume  that the subalgebra $\mathcal{N}$ is a differentiable algebra.   This assumption is satisfied by a large class of examples and applications, as will be discussed later on.    Under this assumption, the real spectrum of $\mathcal{N}$ is an affine differentiable space, denoted by $N = \spec_\R(\mathcal{N})$.  Let us denote by $\phi\colon X \to N$ the canonical projection from $X$ to $N$. There is another natural projection $\eta \colon N \to C$ that makes commutative the following diagram, i.e., such that $\pi = \eta\circ \phi$.
\[
\begin{array}{cc} { \bfig \xymatrix{&X\ar[dr]^{\pi}\ar[r]^{\rho} & N\ar[d]^{\eta}\\   && C&}\efig} \end{array}
\]

Given a subset $\mathcal{R} \subset \mathcal{F}$, by analogy with the definition of the commutant of a subset $\mathcal{S}$ of the Lie algebra $\mathcal{L}$, we can define its commutant with respect to the representation $\rho$ of $\mathcal{L}$, i.e. we introduce $\mathcal{R}'_\rho = \{ x \in \mathcal{L} \mid \rho(x)(f) = 0, \, \forall f \in \mathcal{R} \}$.  This notion is consistent with the definition of the bicommutant of $\mathcal{S} \subset \mathcal{L}$ given in Eq. \eqref{Srho_commutant}.
In particular, if we consider $\mathcal{N} \subset \mathcal{F}$, and $\mathcal{L} = \mathrm{Der}(\mathcal{F})$, $\rho = \mathrm{id}$, then we put $\mathcal{N}_\rho'\equiv\mathcal{N}'$, with
$$
\mathcal{N}' = \{ Y \in \mathrm{Der}(\mathcal{F})  \mid Y(g) = 0, \forall g \in \mathcal{N} \} \, .
$$

\subsection{The flow of a regular derivation on a differentiable algebra}\label{sec:structure_flow}

In the rest of this section, we shall assume that all of the algebras appearing on the discussion are of finite type. Under this assumption, the filtration given in Eq. (\ref{filtr}) necessarily stabilizes, that is:
\[ \exists \,  \nu\in \mathbb{N} \qquad \mathrm{s.t}. \qquad \mathcal{C}^{\nu-1}\neq\mathcal{C}^\nu = \mathcal{C}^{\nu+1} = \mathcal{C}^{\nu+2} = \ldots
\]
The index $\nu$ will be called \textit{the index of nilpotency} of $\Gamma$.

We can construct a generating set for $\mathcal{F}$ as follows.  Let $\{h_\alpha\}$, $\alpha = 1,\ldots, r_1$ be a generating set for $\mathcal{C}$.
The $\mathcal{C}$--modules $\mathcal{N}^k = \mathcal{C}^{k}/\mathcal{C}^{k-1}$ for $k > 1$ are finitely generated. For the quotients $\mathcal{N}^k = \mathcal{C}^{k}/\mathcal{C}^{k-1}$, $k > 1$ we will choose a generating set of the form $g_{\alpha_k}^{(k)} + \mathcal{C}^{k-1}$ ($\alpha_k = 1, \ldots, r_k$), such that $\Gamma (g_{\alpha_k}^{(k)}) = g_{\alpha_k}^{(k-1)}$.   Notice that $\Gamma (\mathcal{C}^k) \subset \mathcal{C}^{k-1}$.  Hence the family $\{g_{1}^{(2)}, \ldots, g_{r_2}^{(2)}, \ldots, g_{1}^{(\nu)},\ldots, g_{r_\nu}^{(\nu)}\}$ is a generating set of $\mathcal{N}$ as a $\mathcal{C}$--module.

\begin{definition}
Given an element $g^{(k)}\in \mathcal{C}^k$, the elements $g^{(k-1)} = \Gamma(g^{(k)}), \ldots,$ $g^{(1)} = \Gamma(g^{(2)}), h = g^{(0)} = \Gamma (g^{(1)})$, will be called the \emph{descendants} of $g^{(k)}$.   The set of the descendants of an element $g \in \mathcal{N}$ will be called the tower of higher-order constants of the motion defined by $g$ and denoted by $T(g)$.

Similarly, given an element $h\in \mathcal{C}$, consider the elements $g \in \mathcal{N}$ such that $\Gamma^{k-1}(g) = h$ for some $k$.   We will call them the \emph{ascendants} of $h$.  Given $h$, the set of its ascendants will be called the tower of higher-order constants of the motion defined by $h$.
\end{definition}
Hence the generating family constructed before decomposes as the union of the towers of higher-order constants of the motion defined by the contants of the motion $h_\alpha$.   See Fig. \ref{diagram1} for a pictorial description of the structure of $\mathcal{N}$.

\begin{figure}[t]
\centering
\includegraphics[width=13cm]{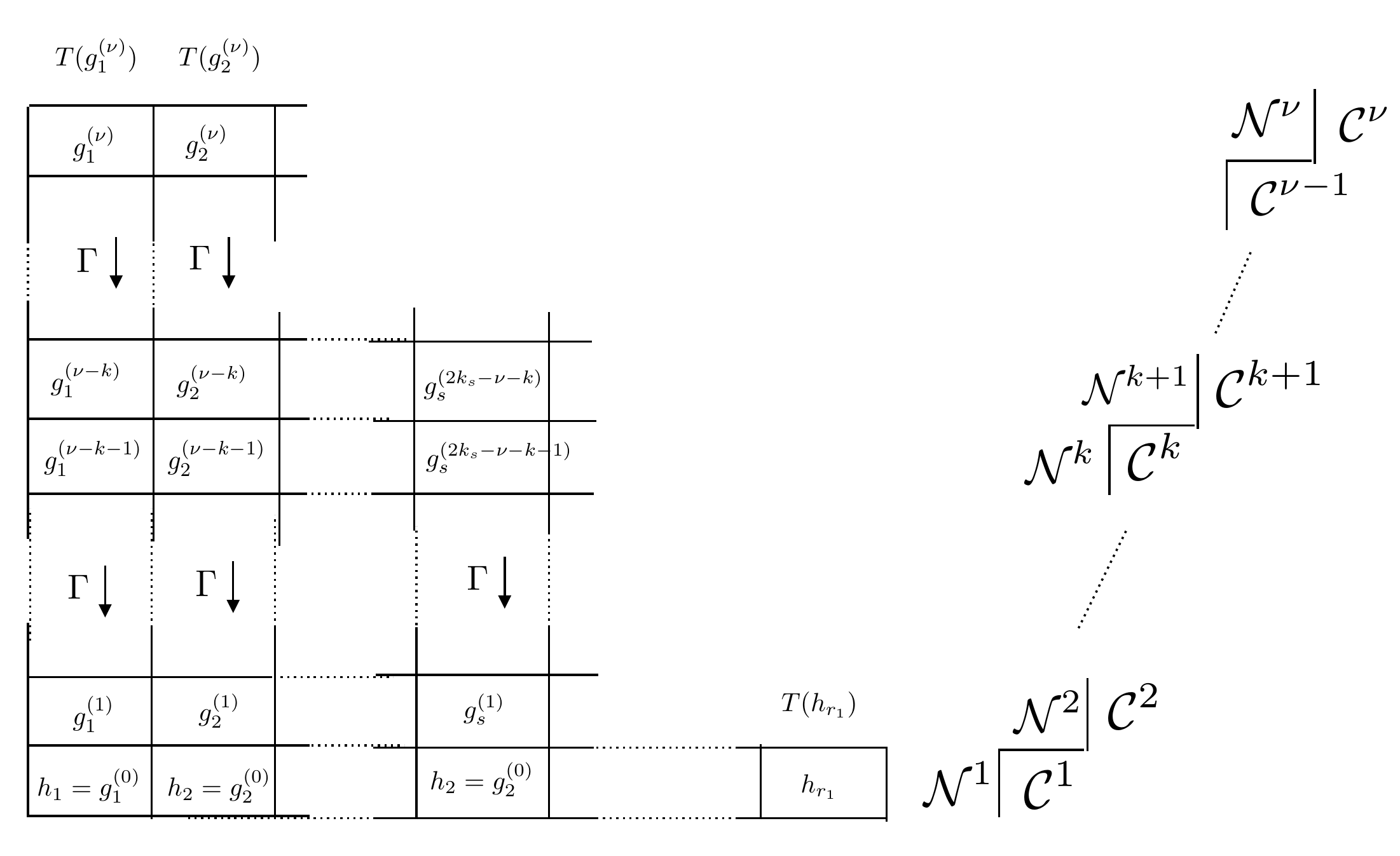}
\caption{The structure of the nilpotent algebra $\mathcal{N}$:  its towers of generalised constants of the motion and the horizontal strata}\label{diagram1}
\end{figure}

Finally, to complete a generating set for $\mathcal{F}$, we will extend the previous family by adding a family of functions $\{f_1, \ldots, f_n\}$. We will call the generating set so constructed a \textit{$\Gamma$--adapted generating set of the algebra $\mathcal{F}$}.

As proved in Lemma \ref{lemma1}, for any point $x\in X=\spec_\R (\mathcal{F})$, the differentials at $x$ of the elements of a $\Gamma$--adapted generating set span the whole cotangent bundle $T_x^{*}X$. Consequently we can extract from them a set of ``local coordinates" for $X$ in a neighborhood of $x$. A local coordinate set around the point $x$ has the form $\{h_1, \ldots, h_{r_1'}$, $g_{1}^{(2)}, \ldots, g_{r'_2}^{(2)}, \ldots, g_{1}^{(\nu)},\ldots, g_{r'_\nu}^{(\nu)}$, $f_1, \ldots, f_{n'}\}$.    We shall call such a local coordinate system a \textit{$\Gamma$--adapted local coordinate system}.

The number of elements $r_1'$, $r_2'$, etc. on each class of an adapted local coordinate system is locally constant, and therefore is constant on each connected component of $X$.  Notice that $r_1' + r_2' + \cdots + r_\nu' + n' = m = \dim X$,

In what follows, we will use the same notation both for the number $r_1$  of generators of $\mathcal{C}$ and for the number $r_1'$ of constants of the motion of an adapted coordinate system, and similarly for $r_2$ and $r_2'$, etc., even though on each class of coordinates of a $\Gamma$--adpated coordinate system $r_1' \leq r_1$, $r'_2 \leq r_2$, etc. With this convention, we have that $r_1 + r_2 + \cdots + r_\nu + n = m = \dim X$.

Due to the fact that $\lie{\Gamma} h_\alpha = 0$, for $\alpha = 1, \ldots, r_1$, then $\varphi^*_t h_\alpha = h_\alpha$ for all $t$.   We also have that $h_\alpha (\varphi_t(x_0)) = h_\alpha (x_0)$ for some given initial data $x_0$; we denote this value by $h_{\alpha,0}$.

Let $g_\alpha^{(2)}$ be a second-order constant of the motion.  Then $\lie{\Gamma}g_\alpha^{(2)} \in \mathcal{C}$. Therefore
\[
\lie{\Gamma}g_\alpha^{(2)} = \nu_\alpha^{(2)}(h_1, \ldots, h_{r_1}).
\]
Then
\[
g_\alpha^{(2)}(\varphi_t(x_0)) = g_\alpha^{(2)}(x_0) + \nu_\alpha^{(2)}(h_{1,0}, \ldots, h_{r_1,0}) t;
\]
consequently, we have:
\begin{equation}\label{evol2}
g_\alpha^{(2)} \circ \varphi_t  =g_\alpha^{(2)} +  \nu_\alpha^{(2)}(h_{1}, \ldots, h_{r_1}) t  .
\end{equation}
For a function $g_\alpha^{(3)} \in \mathcal{C}^3$, we have that $\lie{\Gamma}  g_\alpha^{(3)} \in \mathcal{C}^2$; then the $\mathcal{C}$--module $\mathcal{C}^2$ is finitely generated and
\[
\lie{\Gamma}  g_\alpha^{(3)} = \nu_{\alpha, 1}^{(3)}(h_1, \ldots,h_{r_1}) g_{1}^{(2)} + \cdots +  \nu_{\alpha, r_2}^{(3)}(h_1, \ldots,h_{r_1}) g_{r_2}^{(2)}.
\]
Hence we get
$$g_\alpha^{(3)} \circ \varphi_t =g_\alpha^{(3)} +  \nu_{\alpha, 1}^{(3)}(h_1, \ldots,h_{r_1}) g_{1}^{(2)} + \cdots +  \nu_{\alpha, r_2}^{(3)}(h_1, \ldots,h_{r_1}) g_{r_2}^{(2)}. $$

If we choose now a $\Gamma$-adapted generating set $\{ g_\alpha^{(k)} \}$ of the algebra $\mathcal{F}$,  because of the previous arguments and of the definition of the towers of higher-order constants of the motion, we easy deduce the following relations:

\begin{eqnarray}\label{flow}
h_\alpha \circ \varphi_t &=& h_\alpha  \, , \nonumber \\
g_\alpha^{(2)} \circ \varphi_t &=& g_\alpha^{(2)} +  h_{\alpha} t \,  ,  \nonumber \\
g_\alpha^{(3)} \circ \varphi_t &=& g_\alpha^{(3)} +g_\alpha^{(2)} t + \frac{1}{2}h_{\alpha} t^2 \, , \label{flow_g2} \nonumber  \\
\dots  & & \dots \nonumber \\
g_\alpha^{(k)} \circ \varphi_t &=& g_\alpha^{(k)} +  g_\alpha^{(k-1)}t + \cdots + \frac{1}{k!}h_{\alpha} t^k \, . \end{eqnarray}

\subsection{Strict nilpotent integrability}

We propose several definitions of nilpotent integrability.   The first one is directly tied to the polynomial description of the flow associated with a dynamical system $\Gamma$, according to Eqs. \eqref{flow}.

\begin{definition}
A dynamical system $\Gamma$ on a differentiable algebra $\mathcal{F}$ is strictly nilpotent integrable if there exists a family of higher-order constants of the motion $\mathcal{G} = \{ h_1, \ldots, h_r, g_1, \ldots, g_s \} \subset \mathcal{N}$ that separates points and infinitely near points of $X = \spec_\R (\mathcal{F})$.
\end{definition}

Notice that the algebra generated by $\mathcal{G}$ is dense in $\mathcal{F}$ because it separates points and infinitely near points, hence in such a case $\mathcal{N} = \mathcal{F}$. Then the flow of $\Gamma$ has an explicit description as a polynomial of the form given by Eqs. (\ref{flow}).   In fact, because of the previous discussion, we can choose a local coordinate system around each point of the form
\[
\{h_1, \ldots, h_{r_1}, g_{1}^{(2)}, \ldots, g_{r_2}^{(2)}, \ldots, g_{1}^{(s)},\ldots, g_{r_s}^{(s)}\}
\]
and the flow of the dynamics $\Gamma$ is polynomial in these coordinates. In this situation, we will also say that the dynamics $\Gamma$ is nilpotent integrable of order $s$.

Unfortunately, the above definition of nilpotent integrability, although it provides a direct grasp of its meaning, is hard to check. It would be desirable to have a definition based on the nilpotent subalgebra $\mathcal{N}$ only.  As was already pointed out, constant functions (which are always contained in $\mathcal{N}$) fuzz the picture.  It is more appropriate to consider the algebras of functions vanishing at a given point $x_0 \in M$.   Observe that for any $h_\alpha \in \mathcal{N}$ the function $h_\alpha(x) - h_\alpha (x_0)$ vanishes at $x_0$. Then we have the following main result.

\begin{theorem}
A regular dynamics $\Gamma$ on a differentiable algebra whose real spectrum is formally smooth is nilpotent integrable if and only if  $\mathcal{F}d\mathcal{N} = \Omega^1$. \end{theorem}

\begin{proof}   It is clear that if $\Gamma$ is nilpotent integrable, then $\mathcal{N} = \mathcal{F}$, hence obviously $\mathcal{F} d \mathcal{N} = \Omega^1$.   Conversely, if $\mathcal{F} d \mathcal{N} = \Omega^1$, since $\Omega^1$ defines a locally free sheaf over $X$ of bounded rank (because the spectrum of $\mathcal{F}$ is formally smooth), then there exists a finite family of elements $\beta_1, \ldots, \beta_N$ generating it.   Notice that $\beta_k = f_{kl} dg_{kl}$, $k = 1, \ldots, N$, $l = 1, \ldots r+s$, with $g_{kl} \in \mathcal{N}$.

\noi Consider the family $\mathcal{G}$ formed by the elements $g_{kl}$. Then the family of 1--forms $dg_{kl}$ generate $\Omega^1$ over $\mathcal{F}$. At the same time, the localization $d_xg_{kl}$ of the 1--forms to $x \in X$ generates $T_x^*X$ for any $x \in X = \spec_\R (\mathcal{F})$, hence they separate infinitely near points.

Moreover, let $x \in U$ with $U$ an open set containing a given point $x$ in the real spectrum of $\mathcal{F}$, then the localization $d_Ug_{kl}$ of the family of 1--forms $dg_{kl}$ to the open set $U$ generates the localization $\Omega^1(U)$ of the sheaf of 1--forms on $X$. Consequently, we can extract a family that is independent because the sheaf is locally free.   Let us denote by $dg_{\alpha}$ such a family; then there exists functions $\tilde{g}_{\alpha}$ on some open set $\tilde{U}$ in $\R^n$ where $\mathcal{C}^\infty /\mathcal{J}$ is a representation of the differentiable algebra $\mathcal{F}$, such that $U \subset \tilde{U}$ and $\tilde{g}_\alpha + \mathcal{J} = g_\alpha$.   The spectrum of an affine differentiable space is separated (\cite{Na03}, p. 67). Therefore, given points $x \neq y$ we can find open neighborhoods $U$ and $V$ of $x$ and $y$ respectively such that $U\cap V = \emptyset$. Then the local functions  $\tilde{g}_\alpha$ constructed above on $U$ and $V$  provide a function $g$ on $\mathcal{G}$ separating $x$ and $y$. Hence  $\mathcal{G} \subset \mathcal{N}$ separates points and infinitely near points and $\Gamma$ is nilpotent integrable.
\end{proof}

\subsection{Complete nilpotent integrability}

In the analysis of integrability, a preeminent role is played by \textit{completely integrable systems}.  We have seen already that angle coordinates can be treated as second-order constants of the motion.  However, in general, such coordinates do not define global functions, i.e., they are not defined on the differentiable algebra $\mathcal{F}$.   The explanation for this, as was already pointed out in the Introduction, is that angle variables should be considered in terms of an Abelian algebra of symmetries of $\Gamma$, mainly as the parameters describing the flows of a generating system for it.

In this sense, we introduce the following definition that extends the notion of strict nilpotent integrability and encompasses the notion of completely integrable systems too.

\begin{definition}\label{CI}
A vector field $\Gamma$ is said to be complete nilpotent integrable if there exists an Abelian algebra $\mathfrak{h}_\Gamma$ of symmetries of $\Gamma$, with $\Gamma \in \mathfrak{h}_\Gamma$, such that
$$
\mathfrak{h}_\Gamma \supseteq \mathcal{N}' \, .
$$
\end{definition}

Let us assume, consistently with our previous assumptions, that $\Gamma$ is completely nilpotent integrable and that $\mathfrak{h}_\Gamma$ is finitely generated; we denote by $Y_\mu$ a family of generators.

\begin{theorem}  Let $\Gamma$ be a complete nilpotent integrable dynamical system with a finitely generated Abelian algebra of symmetries $\mathfrak{h}_\Gamma$ such that $\mathcal{N}' \subset \mathfrak{h}_\Gamma$. Then there exists a family of higher-order constants of the motion $g_\mu$ such that
$$
\Gamma = \sum_\mu g_\mu Y_\mu \, ,
$$
and conversely.
\end{theorem}

\begin{proof}   Since $\Gamma \in \mathfrak{h}_\Gamma$, then $\Gamma = \sum_\mu a_\mu Y_\mu$ with $a_\mu\in\mathcal{F}$.   Now, for any derivation $Y$ in the symmetry algebra of $\Gamma$, we have that $[Y, \Gamma ] = 0$ implies that $Y(a_\mu) = 0$.    Hence, the elements $a_\mu$ are annihilated by all derivations in the symmetry algebra of $\Gamma$, in particular by all elements in $\mathfrak{h}_\Gamma$, more specifically by all derivations tangent to the projection $Y \to N$; then they are in $\mathcal{N}$.
\end{proof}

If $\theta_\mu$ are local functions parametrizing the flows of $X_\mu$, then we have that
$$
\Gamma (\theta_\mu) = g_\mu \, ,
$$
and $\Gamma (g_\mu)$ are lower order constants of the motion whose evolution will be described by Eqs. \eqref{flow}.

Clearly, this picture reduces exactly to the case of completely integrable hamiltonian systems in the case when
$$
\mathfrak{h}_\Gamma \supseteq \mathcal{C}' \supseteq \mathcal{N}' \, .
$$

\section{Generalized algebraic reduction of nilpotent integrable systems}\label{sec:reduction_etc}

In this section, we will discuss a generalized reduction procedure, adapted to the framework of nilpotent integrability, that extends the well--known reduction techniques for vector fields on manifolds.   The discussion of generalized algebraic reduction will be done as before on the realm of differentiable algebras.  Later on, particular instances of our construction and the relation with well--known reduction theorems will be discussed.

If we are given a dynamics $\Gamma$ on a differentiable manifold $M$, we can reduce it in two ways: first by restricting it to an invariant submanifold and secondly, by quotienting it with respect to a suitable projection map.  Both procedures are perfectly illustrated in the classical example of a hamiltonian flow with symmetry, first restricted to an invariant submanifold defined by level sets of its constants of the motion and then quotiented with respect to the projection map defined by the existence of cyclic variables.

Hereafter we shall describe these two approaches to the reduction in the setting of differentiable algebras, because it makes the approaches independent on the existence of auxiliary geometrical structures.

\subsection{Restriction to invariant subspaces}\label{sec:res_iv}

Given a differentiable algebra $\mathcal{F}$ with spectrum $X = \spec_\R (\mathcal{F})$, an affine differentiable subspace $Y\subset X$ is determined by a closed ideal $\mathcal{J}$ of the algebra $\mathcal{F}$.  Thus if  $\mathcal{F}$ is the differentiable algebra $\mathcal{F}(M)$ of smooth functions on the manifold $M$ and  $S\subset M$ is a regular submanifold, then the set $\mathcal{J}_S$ of functions vanishing at $S$ defines a closed ideal of $\mathcal{F}(M)$ and $\mathcal{F}(M) / \mathcal{J}_S \cong \mathcal{F}(S)$. In turn, any closed ideal $\mathcal{J}$ of $\mathcal{F}$ defines a differentiable algebra, given by $\mathcal{F} / \mathcal{J}$.   Elements in the quotient space $\mathcal{F} / \mathcal{J}$ will be denoted by $f + \mathcal{J}$ or simply by $[f]$, unless a specific realization in terms of functions on a given set is provided.
\begin{definition}
Given a dynamics $\Gamma$ on $\mathcal{F}$ and a closed ideal $\mathcal{J}$ of $\mathcal{F}$, we shall say that $\mathcal{J}$ is $\Gamma$--invariant if $\Gamma(f) \in \mathcal{J}$ for any $f \in \mathcal{J}$. \end{definition}
In this case, the derivation $\Gamma$ induces a derivation $[\Gamma]$ on the algebra $\mathcal{F} / \mathcal{J}$
by defining $[\Gamma](f + \mathcal{J}) = \Gamma(f) + \mathcal{J}$.  Such a derivation will be said to be \textit{the restriction of $\Gamma$ to} $\mathcal{F}/\mathcal{J}$.
It is clear that if the $\Gamma$--invariant  ideal $\mathcal{J}$ is the ideal of functions vanishing on the submanifold $S$, then $[\Gamma]$ defines a dynamics on the submanifold $S$ which is nothing but the restriction of the vector field $\Gamma$ to $S$.

\subsection{Quotienting by invariant relations}

The algebraic analog of an equivalence relation on a manifold $M$ is a subalgebra of the algebra of functions $\mathcal{F}$.  In fact, let us assume that $\sim$ is an equivalence relation on $M$ such that the quotient space $\widetilde{M} = M/\sim$ is a manifold and the canonical projection $\pi \colon M \to \widetilde{M}$ is a submersion.  Then the pull--back with respect to $\pi$ of the differentiable algebra of smooth functions on $\widetilde{M}$ determines a subalgebra $\pi^*\mathcal{F}(\widetilde{M})$ of $\mathcal{F}(M)$.

Conversely, if $\tilde{\mathcal{F}} \subset \mathcal{F}(M)$ is a subalgebra, we can define an equivalence relation on points of $M$ by declaring that $x \sim y$ iff $f(x) = f(y)$ for all $f \in \tilde{\mathcal{F}}$.    Given a differentiable algebra $\mathcal{F}$, not every subalgebra is a differentiable algebra itself, however this is true for closed subalgebras.  Then if $\tilde{\mathcal{F}}$ is a closed subalgebra of $\mathcal{F}$ and $\tilde{X} = \spec_\R (\tilde{\mathcal{F}})$, $X = \spec_\R(\mathcal{F})$ are the corresponding spectra, there is a natural continuous projection $\pi\colon X \to \tilde{X}$ such that $\pi^* = i$ is the canonical inclusion of $\tilde{\mathcal{F}}$ into $\mathcal{F}$.
As in the situation above of restriction by an invariant subspace, if we are given a dynamics $\Gamma$ and a closed subalgebra $\widetilde{\mathcal{F}}$, we shall say that it is $\Gamma$--invariant if $\Gamma (f)  \in \widetilde{\mathcal{F}}$ for all $f \in \widetilde{\mathcal{F}}$.  In this case, the restriction of $\Gamma$ to $\widetilde{\mathcal{F}}$ defines a derivation of the subalgebra $\widetilde{\mathcal{F}}$.  This restriction, denoted by $\widetilde{\Gamma}$, will be called the quotient of $\Gamma$ with respect to the subalgebra $\widetilde{\mathcal{F}}$.

If  $\widetilde{\mathcal{F}} = \pi^*\mathcal{F}(\widetilde{M})$, as in the example above, then the vector field $\Gamma$ is $\pi$--projectable and induces a vector field $\tilde{\Gamma}$ on $\tilde{M}$ iff $\tilde{\mathcal{F}}$ is $\Gamma$--invariant. The vector field $\tilde{\Gamma}$ is just the quotient of $\Gamma$ with respect to $\pi^*\mathcal{F}(\tilde{M})$.

Notice that a given equivalence relation, determined by a subalgebra $\tilde{\mathcal{F}}$, does not have to be compatible with a given subspace determined by an ideal $\mathcal{J}$.  We shall say that the equivalence relation $\tilde{\mathcal{F}}$ is compatible with the subspace $\mathcal{J}$ if $\mathcal{J}\subset \tilde{\mathcal{F}}$.

\subsection{Generalized reduction}\label{sec:gen_red}

Both procedures discussed above, restriction to a subspace and quotienting by an equivalence relation, can be combined, i.e. we can restrict the dynamics to a given subspace and quotienting it out by an equivalence relation.    In fact we can proceed in two different ways: first we restrict with respect to the ideal $\mathcal{J}$ and then quotient with respect to the closed subalgebra $\tilde{\mathcal{F}}$, or viceversa first we quotient with respect to the subalgebra $\tilde{\mathcal{F}}$ and then we restrict with respect to a given ideal $\mathcal{J}$.    The first procedure leads to a dynamics $[\tilde{\Gamma}]$ on the algebra $\tilde{\mathcal{F}}/(\tilde{\mathcal{F}} \cap \mathcal{J})$ and the second gives a dynamics $\widetilde{[\Gamma]}$ on $(\tilde{\mathcal{F}}+ \mathcal{J})/ \mathcal{J}$. Indeed,  the smallest equivalence relation containing $\mathcal{\tilde{F}}$ and compatible with the subspace defined by $\mathcal{J}$ is the subalgebra $\mathcal{\tilde{F}} + \mathcal{J}$.  The ideal $\mathcal{J}$ is an ideal of $\mathcal{\tilde{F}} + \mathcal{J}$, hence  $\tilde{\mathcal{F}}+ \mathcal{J}/ \mathcal{J}$ is an algebra.   The later reduction in a geometrical form can be traced back to the Marsden--Weinstein reduction theorem in the realm of Hamiltonian systems with symmetry \cite{Ma74}; the former, in an algebraic form, traces back to Grabowski--Marmo reduction theorem for Poisson manifolds \cite{Ga94}.   Consequently, we will call the first reduced dynamics $[\tilde{\Gamma}]$ the geometric Marsden--Weinstein reduction and the second one $\widetilde{[\Gamma]}$, the algebraic Grabowski--Marmo reduction.
Both reduction schemes lead to the same dynamics as is proven in the following theorem.

\begin{theorem}\label{en_red_iso}   Let $\mathcal{F}$ be a differentiable algebra with a unit element and $\Gamma$ be a derivation.  Let $\mathcal{J}$ be a closed $\Gamma$--invariant ideal and $\mathcal{\tilde{F}}$ be a closed $\Gamma$--invariant differentiable subalgebra of $\mathcal{F}$.

The reduced subalgebras $\tilde{\mathcal{F}}/\tilde{\mathcal{J}} \cap \mathcal{J}$ and $\tilde{\mathcal{F}}+ \mathcal{J}/ \mathcal{J}$ are isomorphic differentiable algebras.   Moreover, the reduced dynamics are transformed into each other by the previous isomorphism, that is $[\tilde{\Gamma}] \cong\widetilde{[\Gamma]}$.

\end{theorem}

\begin{proof}  The isomorphism among the algebras $\mathcal{J}\cap \tilde{\mathcal{F}}/\mathcal{\tilde{F}}$ and $\mathcal{\tilde{F}} + \mathcal{J} /\mathcal{J}$ is provided by the second isomorphy theorem. We introduce the map
$$ \Phi (\tilde{f}) = \tilde{f} + \mathcal{J} ,$$
whose kernel is  $ \mathcal{J}\cap \tilde{\mathcal{F}}$; the first isomorphy theorem provides the conclusion we were looking for.  We will denote the isomorphism so constructed by $[\Phi]$.

Finally, let us observe that
$$[\Phi]_*(\tilde{[\Gamma]})(\tilde{f}+ \mathcal{J}) = \tilde{\Gamma}(\tilde{f} + \mathcal{J}\cap \tilde{\mathcal{F}}) = \tilde{\Gamma}(\tilde{f}) + \mathcal{J}\cap \tilde{\mathcal{F}} = [\tilde{\Gamma}](\tilde{f} + \mathcal{J}\cap \tilde{\mathcal{F}} ) .$$

\end{proof}

In what follows, since the algebraic Marsden-Weinstein and Grabowski-Marmo reduced algebras are isomorphic as well as the corresponding reduced dynamics, any of the isomorphic spaces and dynamics obtained in this way will be denoted by $\mathcal{F}\reduction_{\mathcal{J}} \mathcal{\tilde{F}}$ and the induced dynamics by $\Gamma_{\mathcal{J},\tilde{\mathcal{F}}}$.

\subsection{Reduction of nilpotent integrable systems}\label{sec:reduc_nilpotent}

We shall discuss now what happens to the integrability properties of  a system under algebraic reduction, namely how do the constants of the motion and the nilpotent algebras behave under algebraic reduction.   We shall focus only on the simpler situations, namely the extension to nilpotent integrable systems of Jacobi's procedure of elimination of nodes, Poincar\'e's reduction of order, or reduction by cyclic variables in classical mechanics \cite{Sm70}.

Consider a $\Gamma$-invariant family $\mathcal{G} \subset \mathcal{N}$.  If $g\in \mathcal{G}$, then $\Gamma(g) \in \mathcal{G}$, i.e., all descendants of the higher-order constant of the motion $g$ must belong to $\mathcal{G}$.  Thus the tower $T(g)$ starting at $g$ must be contained in $\mathcal{G}$.   Hence a $\Gamma$-invariant family $\mathcal{G}$ in the nilpotent algebra is the union of towers $T(g)$ of higher-order constants of the motion.   Because of this, we will also say that the family of higher-order constants of the motion $\mathcal{G}$ is nilpotent complete.

Let $\mathcal{J}$ be a closed ideal generated by a $\Gamma$-invariant  family of higher-order constants of the motion $\mathcal{G} = \{ g_\alpha \}$.   It is clear that such an ideal is $\Gamma$-invariant because the tower $T(g_\alpha) \subset \mathcal{G} $, i.e. $\Gamma (g_\alpha) \in \mathcal{G}$ for any $g_\alpha \in \mathcal{G}$.
The restricted algebra defined by $\mathcal{J}$ is given by $\mathcal{F}/\mathcal{J}$.

Let us consider a subalgebra $\widetilde{\mathcal{F}} \subset \mathcal{F}$ defined by a family of Abelian symmetries $\mathfrak{s} \subset \mathfrak{h}_\Gamma$, i.e., $f\in \widetilde{\mathcal{F}}$ if $X_\mu (f) = 0$ for all $X_\mu\in\mathfrak{s}$.   We shall denote the previous subalgebra as $\widetilde{\mathcal{F}} = \mathcal{F}^\mathfrak{s}$.

The following statement follows easily from the theory developed above.

\begin{theorem}  Let $\Gamma$ be a nilpotent integrable system on the differentiable algebra $\mathcal{F}$ with Cartan algebra $\mathfrak{h}_\Gamma$.   Let $\mathcal{J}_\mathcal{G} = \mathcal{F}\mathcal{G}$ be an ideal generated by a nilpotent complete family of higher-order constants of the motion $\mathcal{G}$, and $\mathfrak{s} \subset \mathfrak{h}_\Gamma$ a family of Abelian symmetries of $\Gamma$.  The reduced dynamical system $\Gamma_{\mathcal{J},\tilde{\mathcal{F}}}$ defined on the reduced algebra $\mathcal{F}\reduction_{\mathcal{J}_\mathcal{G}} \mathcal{F}^\mathfrak{s}$ defined by the ideal $\mathcal{J}_\mathcal{G}$ and the subalgebra $\mathcal{F}^\mathfrak{s}$ is nilpotent integrable.
\end{theorem}

\begin{proof}   Due to the fact that we are in the generalized reduction scheme described in previous Sections \ref{sec:res_iv}-\ref{sec:gen_red}, we may use Theorem  \ref{en_red_iso}.

We check first that the reduction of $\Gamma$ to the invariant subspace defined by the ideal $\mathcal{J}_\mathcal{G}$ is nilpotent integrable.
It is easy to ascertain that the nilpotent algebra $\mathcal{N}_{\mathcal{J}_\mathcal{G}}$ of the restriction $[\Gamma]$ of $\Gamma$ with respect to the ideal $\mathcal{J}_\mathcal{G}$ contains the nilpotent algebra $\mathcal{N}$ restricted to $\mathcal{F}/\mathcal{J}_\mathcal{G}$.  Then
\begin{equation}\label{N'_restric}
\mathcal{N}_{\mathcal{J}_\mathcal{G}}' \subset (\mathcal{N}/(\mathcal{N}\cap \mathcal{J}_\mathcal{G})' \, .
\end{equation}
On the other hand, maximal Abelian subalgebras $\mathfrak{h}_{[\Gamma]}$ of the restricted derivation $[\Gamma ]$ are larger than the restricted Abelian subalgebra $[\mathfrak{h}_\Gamma]$.  Since the system is nilpotent integrable, we have $\mathcal{N}' \subset \mathfrak{h}_\Gamma$. Then, due to Eq. \eqref{N'_restric}, we get
$$(\mathcal{N}/(\mathcal{N}\cap \mathcal{J}_\mathcal{G})'  \subset [\mathfrak{h}_\Gamma] \, .$$
Hence we conclude that
$$
\mathcal{N}_{\mathcal{J}_\mathcal{G}}' \subset [\mathfrak{h}_\Gamma] \subset \mathfrak{h}_{[\Gamma]} \, ,
$$
and the restricted system $[\Gamma]$ is nilpotent integrable.

Consider now the quotient system $\widetilde{\Gamma}$ defined on the subalgebra $\widetilde{F} = \mathcal{F}^\mathfrak{s}$.    We have that the nilpotent algebra $\widetilde{\mathcal{N}}$  of $\widetilde{\Gamma}$ is just $\mathcal{N}\cap \mathcal{F}^\mathfrak{s}$, i.e., $ \widetilde{\Gamma} = \widetilde{\Gamma}^\mathfrak{s} = \mathcal{N} \cap \widetilde{\mathcal{F}}$. At the same time, because $\Gamma$ is nilpotent integrable, we deduce that $\mathcal{N}' \subset \mathfrak{h}_\Gamma$, hence
$$
\widetilde{\mathcal{N}}' = \mathcal{N}' + \widetilde{\mathcal{F}}' = \mathcal{N}' + \mathfrak{s}'' \, .
$$
Observe that since $\widetilde{\mathcal{F}} = \mathcal{F}^\frak{s}$, then $\widetilde{\mathcal{F}}$ is the invariant set $\mathfrak{s}'$; therefore $\widetilde{\mathcal{F}}' = \mathfrak{s}''$.  But because of the stability Lemma \ref{stability2}, $\mathfrak{s}$ being Abelian implies that $\mathfrak{s}''$ is Abelian too and $\mathfrak{s} \subset \mathfrak{s}''$.   Now, $\mathfrak{h}_\Gamma$ is a maximal Abelian algebra of $\Gamma$, so $\mathfrak{s} \subset \mathfrak{h}_\Gamma$ implies that $\mathfrak{s}'' \subset \mathfrak{h}_\Gamma$.  We conclude that
$$
\widetilde{\mathcal{N}}' \subset \mathfrak{h}_\Gamma = \mathfrak{h}_{\widetilde{\Gamma}} \, ,
$$
and the system is nilpotent integrable.
\end{proof}

\section{Some explicit classes of nilpotent systems. The higher-order Calogero-Moser systems}\label{sec:examples}

\subsection{Linear systems}  A simple instance of derivations is provided by derivations on $\mathcal{C}^\infty (\R^n)$ which are linear with respect to
the linear structure on $\R^n$.

If $\Gamma$ is a derivation on $\mathcal{C}^\infty (\R^n)$, it is linear if $[\Gamma , \Delta ]= 0$, where $\Delta$ is the derivation whose flow is given by $\varphi_t (x) = e^t x$.  In global linear coordinates $x^i$ on $\R^n$, the derivation $\Delta$ is the dilation vector field $\Delta = x^i \partial / \partial x^i$.
If $\Gamma$ is linear, there exists a linear  map $A \colon \R^n \to \R^n$ such that $\Gamma = X_A = A_i^j x^i \partial /\partial x^j$, where $X_A$ is the derivation whose flow is given by $\psi_t (x) = e^{tA}(x)$.    The algebra $\mathcal{C}^\infty(\R^n)$ is generated by any collection of linear functions $f_\alpha (x) = \langle\alpha,x\rangle$, $\alpha \in (\R^n)^*$, containing a generating set $\alpha_1, \ldots, \alpha_n$ for $(\R^n)^*$.

A simple computation shows that $\lie{\Gamma} f_\alpha = f_{A^*\alpha}$, where $A^*$ is the adjoint map to $A$.    Hence the algebra of constants of the motion $\mathcal{C}$ for $\Gamma$ is generated by the linear functions $f_\alpha$ where $\alpha \in \ker A^*$.    One can see that the real spectrum of $\mathcal{C}$ is given by $\ker A$ and the natural projection is given by $\pi \colon \R^n \to \ker A$, defined as $\pi (x) \alpha = \langle \alpha, x \rangle$, for all $\alpha \in \ker A^* = (\ker A)^*$.  Alternatively, we can choose a basis $\alpha_1, \ldots, \alpha_r$ of $\ker A^*$. Then the map $\pi$ can also be written as $\pi = f_{\alpha_i} v^i $ where $v^i$ is the dual basis to $\alpha_i$.

In a similar way we obtain the nilpotent algebra $\mathcal{N}$ of $\Gamma$ as the subalgebra generated by the linear functions $f_\alpha$ where $(A^*)^k \alpha = (A^k)^* \alpha = 0 $ for some $k$.  Thus if $A$ is nilpotent, then $\mathcal{N} = \mathcal{C}^\infty(\R^n)$ and the system is strictly nilpotent integrable (and it  possesses a generating set of linear higher-order constants of the motion).

In the case of an arbitrary linear system $\Gamma = X_A$, we may use the Jordan decomposition of $A$ in its semisimple and nilpotent part, i.e., $A = S + N$ where $S$ is diagonalizable, $N$ is nilpotent and $[S,N] = 0$.  Then, the algebra generated by the powers of $S$ defines a Cartan subalgebra for $\Gamma$ and the system is trivially nilpotent integrable.

%%%%%%%%%%%%%%%%%%%%%%%%
%%%%%%%%%%%%%%%%%%%%%%%%

\subsection{Reduction of higher-order free systems}

\subsubsection{Reduction of the free system of order two: The Calogero-Moser system}

As in the Introduction, let $\mathbf{r}(t)=(x(t),y(t),z(t))$ be the coordinates of a particle moving freely in $\mathbf{R}^3$, i.e., satisfying $\ddot{\mathbf{r}}= 0$.
Then we may define a $2\times 2$ real symmetric matrix
$$
X(t)=\begin{pmatrix}
x + z & y \\ y & x - z\end{pmatrix}
$$
that satisfies the equivalent system
\begin{equation}
\ddot{X}=0 \, .
\end{equation}

As we already discussed,  we can reduce this system in order to obtain new dynamical systems exhibiting non-trivial potential terms.   We can formulate abstractly this situation by considering the tangent space $M = T\mathcal{H}$ to the space of $2\times 2$ real symmetric matrices $\mathcal{H}$.  The points in $M$ will be denoted by $(X,V)$, $X,V \in \mathcal{H}$.   Consider now the second order system on $\mathcal{H}$ defined by $\ddot{X} = 0$. Its flow is given by $X(t) = X_0 + V_0t$, $X_0,V_0 \in \mathcal{H}$.   Notice that the vector field corresponding to such system is
\begin{equation}\label{free_motion}
\Gamma = V \frac{\partial}{\partial X} \, ,
\end{equation}
with the obvious matrix multiplication product understood.

We may write the matrix $X \in \mathcal{H}$ as:
\begin{equation}\label{defX}
X(t)=   x \mathbb{I} + y\sigma_1 + z \sigma_3 \, ,
\end{equation}
in terms of the matrices $\sigma_k$, $k = 1,2,3$ \footnote{Notice that $\sigma_2$ here is $i$-times the second Pauli matrix.}:
$$
\sigma_1 = \begin{pmatrix} 0 &1 \\ 1 & 0 \end{pmatrix}\, , \quad  \sigma_2 = \begin{pmatrix} 0 &1 \\ -1 & 0 \end{pmatrix} \, , \quad \sigma_3 = \begin{pmatrix} 1 & 0 \\ 0 & -1 \end{pmatrix} \, .
$$
The group $U(1) = SO(2)$, realized by matrices
\begin{equation*}
G =\begin{pmatrix} \cos \varphi & \sin \varphi \\ -\sin \varphi & \cos \varphi\end{pmatrix} \, ,
\end{equation*}
acts by conjugation on $\mathcal{H}$, $X \mapsto GXG^{-1}$, and transforms any matrix $X$ of the form \eqref{defX} into
a diagonal matrix:
\begin{equation}\label{GQG}
X = G Q G^{-1} = G \begin{pmatrix} q_1 & 0\\ 0Ê& q_2 \end{pmatrix}G^{-1} \, .
\end{equation}
We can also write the previous relation as
$$
X = G M^{(0)}G^{-1} \, , \qquad \mathrm{with} \qquad M^{(0)} = Q\,,
$$
for later convenience.
Let $q= q_2 - q_1$ denote the difference between the eigenvalues of $Q$.
The infinitesimal generator $\tau$ of $U(1)$ determines a basis for its Lie algebra:
$$
\tau = G^{-1} \dot{G} = \dot{\varphi} \sigma_2 \, .
$$

Below we list some useful commutation relations:
\begin{eqnarray}
&& [\sigma_1, \sigma_2] = -2\sigma_3\, , \quad [\sigma_2, \sigma_3] = -2\sigma_1 \, ,\label{rel_1}  \\
&& [\sigma_1, Q] = q \sigma_2\, , \quad [\sigma_2, Q] = q \sigma_1 \, , \quad [\sigma_2, \dot{Q}] = \dot{q} \sigma_1 \, ,  \label{rel_2} \\
&& [\sigma_2, G] = 0 \, ,\qquad [\sigma_1, G] = -2\sigma_3  \sin \varphi \, , \qquad [\sigma_3, G] = 2\sigma_1 \sin \varphi  \, . \label{rel_3} \\
&& [\tau, Q] = q\dot{\varphi} \sigma_1 \, , \quad [\tau, \dot{Q}] = \dot{q}\dot{\varphi} \sigma_1 \, , \quad [\tau,\sigma_1] = 2\dot{\varphi} \sigma_3 \, . \label{rel_4}
\end{eqnarray}

%\noi The explicit relation between the variables $x_i$  and $q_1,q_2,\varphi$ is:
%\begin{eqnarray}
%x_1=&\frac{1}{2} (q_1+q_2)\\
%x_2=&-\frac12 (q_1-q_2) \sin2\varphi\\
%x_3=&\frac{1}{2} (q_1-q_2)\cos 2 \varphi
%\end{eqnarray}

%Let us compute the second derivative of this matrix $Q(t)$. First observe that
%\begin{equation}
%\dot{X}= \dot{G}QG^{-1}+G\dot{Q}G^{-1}-GQG^{-1}\dot{G}G^{-1}=G \big(\dot{Q}+[G^{-1}\dot{G},Q]\big)G^{-1}
%\end{equation}
%and, since
%\begin{equation}
%G^{-1}\dot{G}= \sigma\dot{\varphi}
%,\quad
%G^{-1}\ddot{G}= \sigma\ddot{\varphi}-\dot{\varphi}^2
%\end{equation}

Computing the first derivative of a curve $X(t)$ using the factorizacion Eq. \eqref{GQG}, leads to:
$$
\dot{X}= G\big(\dot{Q}+[\tau,Q]\big)G^{-1}  = GM^{(1)} G^{-1}\, .
$$
with
\begin{equation}\label{M1}
M^{(1)} = \dot{Q} + [\tau, Q] = \dot{M}^{(0)} + [\tau, M^{(0)}] \, .
\end{equation}
Then, we get immediately:
\begin{equation}\label{Xdotdot}
\ddot{X} = G\big(\dot{M}^{(1)}+[\tau,M^{(1)}]\big)G^{-1} = GM^{(2)} G^{-1}
\end{equation}
with
\begin{equation}\label{M2}
M^{(2)} = \dot{M}^{(1)} + [\tau, M^{(1)}]  \, .
\end{equation}
Iterating the computation, we get
\begin{equation}\label{recursion}
X^{(k+1)} = GM^{(k+1)}G^{-1} \, , \qquad M^{(k+1)} = \dot{M}^{(k)} + [\tau, M^{(k)}] \, , \quad k \geq 0 \, .
\end{equation}

Notice that Eq. \eqref{Xdotdot}, exhibits explicitly the $U(1)$ symmetry of the system $\ddot{X} = 0$ and provides its equivalent formulation
$$
 \dot{M}^{(1)} = - [\tau, M^{(1)}] \, ,
$$
which is also a Lax representation for it.
Equivalently, substituting the expression for $M^{(1)}$ given by Eq. \eqref{M1} we get
$$
\ddot{Q} = - 2q\dot{\varphi}\sigma_3 + (\ddot{\varphi} + 2\dot{q}\dot{\varphi})\sigma_1 \,.
$$
This relation, by working the commutators out (using Eqs. \eqref{rel_2}) and taking into account that the matrix $Q$ is diagonal, leads to the system
%\begin{eqnarray*}
%\ddot{X} &= &  G
%\big( G^{-1}\dot{G}\big(\dot{Q}+[\sigma_2,Q]\dot{\varphi}\big)
%+ \big(\ddot{Q}+[\sigma_2,\dot{Q}]\dot{\varphi}+[\sigma_2,Q]\ddot{\varphi}\big)
%- \big(\dot{Q}+[\sigma_2,Q]\dot{\varphi}\big)G^{-1}\dot{G}\big)G^{-1}
%\\  & = &G \big(\ddot{Q} +[\sigma_2,Q]\ddot{\varphi}
%+ 2[\sigma_2,\dot{Q}]\dot{\varphi}+ [\sigma_2,[\sigma_2,Q] ]\dot{\varphi}^2\big)G^{-1}.
%\end{eqnarray*}
%\begin{equation}\label{eqQ}
%\ddot{Q}=-[\sigma_2,Q]\ddot{\varphi}
%- 2[\sigma_2,\dot{Q}]\dot{\varphi}- [\sigma_2,[\sigma_2,Q] ]\dot{\varphi}^2 \, .
%\end{equation}
%Using Eq. \eqref{rel_2}, we get:
%\begin{equation}
%\ddot{Q}=\big(q\ddot{\varphi}
%+2\dot{q}\dot{\varphi}\big)\sigma_1+2q\dot{\varphi}^2\sigma_3,
%\end{equation}
%where $q=q_1-q_2$, which is equivalent to the equations:
\begin{eqnarray*}
\begin{cases}
\ddot{Q}=&-2q\dot{\varphi}^2\sigma_3\label{eq1}\\
0=&q\ddot{\varphi} +2\dot{q}\dot{\varphi}\label{eq2} \, .
\end{cases}
\end{eqnarray*}
These three equations can be explicitly written as (using $\tilde{q}=q_1+q_2$):
\begin{equation*}
\begin{cases}
\ddot{\tilde{q}}=&0 \\
\ddot{q} =&4q\dot{\varphi}^2 \\[5pt]
\ddot{\varphi}=&- \dfrac{2\dot{q}}{q} \dot{\varphi}. \label{39}
\end{cases}
\end{equation*}
The third equation can be easily reduced by a quadrature to
$$
\dot{\varphi}=\frac{C}{q^2} \, ,
$$
where $C$ is a constant.  In fact, this equation can be viewed as the construction of an invariant $C:=q^2\dot{\varphi}$, and the previous system becomes:
\begin{equation}
\begin{cases}
\ddot{\tilde{q}}=&0 \\
\ddot{q} =& \dfrac{4C^2}{q^3} \\[10pt]
\ddot{\varphi}=&-\dfrac{2\dot{q}}{q}\dot{\varphi}=-\dfrac{2C\dot{q}}{q^3} \, .
\end{cases}
\end{equation}

However, thinking in the higher-order situation, we prefer to proceed emphasizing the role of constants of the motion and the generalized reduction scheme discussed in Section \ref{sec:reduc_nilpotent}.
Clearly the quantity $L=[X,\dot{X}]$, representing the angular momentum, is a constant of the motion:
$$
\dot{L}=\frac{\mathrm{d}}{\mathrm{d}t}[X,\dot{X}]=[X,\ddot{X}]=0.
$$
In coordinates $(q,\varphi)$ we get:
\begin{equation}\label{L}
L=[X,\dot{X}]= - q^2\dot{\varphi}\sigma_2=-\ell\sigma_2 \, ,\quad \ell=\frac12\mathrm{tr}(L\sigma_2) \, ,
\end{equation}
where $\ell$ is the third component of the angular momentum, and
\begin{eqnarray}
L\sigma_2=q^2\dot{\varphi}I_2,\quad \dot{\varphi}=\frac{1}{2q^2} \mathrm{tr}(L\sigma_2)=\frac{\ell}{q^2}.
\end{eqnarray}
We may use any level set of the function $\ell$ to define an invariant subspace.  The restriction to it of the vector field $\Gamma$, Eq. \eqref{eq1} reads
\begin{eqnarray*}
\begin{cases}
\ddot{Q}=& -\dfrac{2\ell^2}{q^3} \sigma_3 \\
\ddot{\varphi}=&-\dfrac{2\ell\dot{q}}{q^3} \, .
\end{cases}
\end{eqnarray*}
It shows that the vector field $\partial /\partial \varphi$ defines an Abelian symmetry algebra of $\Gamma$. Quotienting with respect to it we finally get, in terms of the coordinates $q_1,q_2$, the classical Calogero-Moser system:
\begin{equation}
\begin{cases}
\ddot{q}_1=& \displaystyle{\frac{2\ell^2}{(q_1-q_2)^3}},\\[15pt]
\ddot{q}_2=& \displaystyle{\frac{2\ell^2}{(q_2-q_1)^3}}.
\end{cases}
\end{equation}

%\subsection{A Hamiltonian point of view}
%
%The Hamiltonian of the system $\ddot{x}_i=0$, $i=1,2,3$ is
%\begin{equation}
%H_0=\frac{1}{2}(p_1^2+p_2^2+p_3^2)
%\end{equation}
%After the change of coordinates:
%\begin{eqnarray*}
%x_1&=&\frac{1}{2} (q_1+q_2)\\
%x_2&=&-\frac12 (q_1-q_2) \sin2\varphi\\
%x_3&= &\frac{1}{2} (q_1-q_2)\cos 2 \varphi
%\end{eqnarray*}
%\begin{eqnarray*}
%p_{1}&=&P_1+P_2\\
%p_{2}&=&-(P_1-P_2)\sin 2\varphi-\frac{p_{\varphi}}{q_1-q_2}\cos2\varphi\\
%p_{3}&=&(P_1-P_2)\cos 2\varphi-\frac{p_{\varphi}}{q_1-q_2}\sin2\varphi\\
%\end{eqnarray*}
%it becomes
%\begin{equation}
%H=P_1^2+P_2^2+\frac{p_{\varphi}^2}{(q_1-q_2)^2}
%\end{equation}
%and, since it does not depend on $\varphi$, we can take $p_{\varphi}=2\ell $ (constant) and obtain the 2-dimensional Hamiltonian of Calogero-Moser:
%\begin{equation}
%H=P_1^2+P_2^2+ \frac{4\ell^2}{(q_1-q_2)^2}
%\end{equation}

\subsubsection{Reduction of the system $\dddot{X}=0$:  The third-order Calogero-Moser system}
We will discuss now various reductions of the third order system $\dddot{X} = 0$ that represents an uniformly accelerated system.

By analogy with the previous discussion, the system is defined in the space $M = T^2\mathcal{H}$, whose points will be denoted by $(X,V,A)$, $X,V,A \in \mathcal{H}$.  The vector field associated with it is
$$
\Gamma = V \frac{\partial}{\partial X} + A \frac{\partial}{\partial V} \, .
$$
We may certainly perform a reduction of the system by means of the constant of the motion $\ddot{X}= A =$ const.  However, we prefer to perform a different reduction that would mimic the angular reduction discussed before for the system $\ddot{X}= 0$.  Besides, in doing so we will show how to proceed to treat systematically higher order `free' systems of the form $X^{(k)} = 0$, $k > 3$.  Notice that the system $\Gamma$ is trivially nilpotent integrable with flow $X(t) = X_0 + V_0 t + \frac12 A t^2$.

Keeping with the same notation as before, the recursion equations \eqref{recursion} provide
$$
\dddot{X} = G M^{(3)} G^{-1} \, \qquad \mathrm{with} \quad  M^{(3)} = \dot{M}^{(2)} + [\tau, M^{(2)}] \,.
$$
Thus, the $U(1)$-invariance of the system leads to the reduced system
\begin{equation}\label{dotM2}
\dot{M}^{(2)} = - [\tau, M^{(2)}]  \, ,
\end{equation}
with $M^{(2)}$ given by Eq. \eqref{M2}. Before computing the explicit expression of the system \eqref{dotM2},
we will discuss the role of higher order constants of the motion associated to the angular momentum in this context.

The angular momentum $L=[X,\dot{X}]$ is not anymore a first order constant of the motion:
$$
\dot{L}=\frac{\mathrm{d}}{\mathrm{d}t}[X,\dot{X}]=[X,\ddot{X}]\neq 0 \, .
$$
However,
\begin{equation}
\ddot{L}=\frac{\mathrm{d}^2}{\mathrm{d}t^2}[X,\dot{X}] = [\dot{X},\ddot{X}] \, ,
\end{equation}
and, finally,
\begin{equation}
\dddot{L}=\frac{\mathrm{d}^3}{\mathrm{d}t^3}[X,\dot{X}] =0 \, ,
\end{equation}
thus $L$ is a third-order constant of the motion.   The quantity $L$ defines the tower of higher-order constants of the motion $L^{(3)} = L, L^{(2)} = \dot{L}$ and $L^{(1)} = \ddot{L}$, where the upper index indicates the order of the constant of the motion as in Section \ref{sec:structure_flow}.  More generally, if we consider the equation $X^{(k)}=0$, the quantity $L$ will be a constant of the motion of order $2k-3$ (except for $n=1$ where $L=0$).

As we know from the previous discussion, Eq. \eqref{L} gives
$L=-q^2\dot{\varphi} \sigma_2 $.  Then it is easy to ascertain that all the derivatives $L^{(k)}$ are proportional to $\sigma_2$. Then, multiplying by $\sigma_2$ we can consider the quantity $L^{(k)}\sigma_2$, which is a multiple of the identity. Taking their traces we obtain the scalar functions
\begin{eqnarray}
\ell_1 = \frac12 \mathrm{Tr}(L\sigma_2)&=& q^2\dot{\varphi}  \, ,
\\
\ell_2 = \dot{\ell_1} = \frac12 \mathrm{Tr}(\dot{L}\sigma_2)&=&2q\dot{q}\dot{\varphi}+ q^2\ddot{\varphi}  \, ,
\\
\ell_3 = \dot{\ell_2} = \frac12 \mathrm{Tr}(\ddot{L}\sigma_2)&=&
2\dot{q}^2\dot{\varphi}
+2q\ddot{q}\dot{\varphi}
+4q \dot{q}\ddot{\varphi}
+q^2\dddot{\varphi}.   \label{angm}
\end{eqnarray}
As in the second-order case, $\ell_3$ is the third component of the generalized `angular momentum $L^{(1)} = \ddot{L}$ and $\dot{\ell_3} = 0$. We can integrate the relations $\dot{\ell_2}= \ell_3$, and $\dot{\ell_1} = \ell_2$ to obtain:
\begin{eqnarray*}
\ell_1 &=& \frac12 \ell_{30} t^2+\ell_{20} t+\ell_{10} \, ,
\\
\ell_2 &=& \ell_{30} t + \ell_{20} \, , \\
\ell_3 &=& \ell_{30} \, ,
\end{eqnarray*}
with $\ell_{10}$, $\ell_{20}$ and $\ell_{30}$, are initial values for the quantities $\ell_1$, $\ell_2$ and $\ell_3$ respectively. At this stage, we can use the tower of angular constants of the motion $\ell_k$ to compute $M^{(k)}$.
First observe that
$$
M^{(1)} = \dot{Q} + [\tau, Q] = \dot{Q} +q\dot{\varphi} \sigma_1 = \dot{Q} + \frac{\ell_1}{q}\sigma_1 \, ,
$$
and
$$
M^{(2)} = \dot{M}^{(1)} + [\tau, M^{(1)}] = \ddot{Q} + 2 \frac{\ell_1^2}{q^3} \sigma_3  +  \frac{\ell_2}{q}\sigma_1  \, .
$$
It can be written, for purposes that will become clear in the subsequent computations, in the form
$$
M^{(2)}  = \ddot{Q} + M^{(2)}_3 \sigma_3 + M^{(2)}_1 \sigma_1 \, ,
$$
with $M^{(2)}_3 = 2 \ell_1^2 / q^3$ and $M^{(2)}_1 = \ell_2 / q$.
Hence we get for $M^{(3)}$ the expression
\begin{eqnarray}
M^{(3)} &=& \dot{M}^{(2)} + [\tau, M^{(2)}] \label{M3}  \\ &=& \dddot{Q} + \left( \dot{M}^{(2)}_3+ 2\dot{\varphi}M^{(2)}_1 \right)\sigma_3 + \left( \dot{\varphi}\ddot{q} +  \dot{M}^{(2)}_1  - 2\dot{\varphi} M^{(2)}_3 \right) \sigma_1 \nonumber  \\
\label{M3}  &= & \dddot{Q} + M^{(3)}_3 \sigma_3 + M^{(3)}_1 \sigma_1 \, , \nonumber \end{eqnarray}
with
\begin{eqnarray*}
M^{(3)}_3 &=& \dot{M}^{(2)}_3+ 2\dot{\varphi}M^{(2)}_1 = 6 \frac{\ell_1\ell_2}{q^3} - 6\frac{\ell_1^2\dot{q}}{q^4} \, , \label{M33}\\
M^{(3)}_1 &=& \dot{\varphi}\ddot{q} + \dot{M}^{(2)}_1  -2\dot{\varphi} M^{(2)}_3 =\frac{\ell_3}{q} + \frac{\ell_1 \ddot{q}- \ell_2\dot{q}}{q^2} - 4 \frac{\ell_1^3}{q^5} \, . \label{M31}
\end{eqnarray*}
Hence, as $Q$ is diagonal, the relation $M^{(3)} = 0$ provides the system of equations
\begin{equation}
\left\{  \begin{array}{l}
\dddot{Q} = -\displaystyle{6 \frac{\ell_1\ell_2 q - \ell_1\dot{q}}{q^4}}\sigma_3 \, ,\\
0 =  \displaystyle{\frac{\ell_3}{q} + \frac{\ell_1 \ddot{q} - \ell_2\dot{q}}{q^2} - 4 \frac{\ell_1^3}{q^5} } \, .
 \end{array}\right.
\end{equation}
Thus we have obtained from the generalized angular momentum reduction of the third-order free system, a new integrable system.
\begin{definition}
The system of equations in $\mathbb{R}^2$
\begin{equation}\label{3rdCM}
\left\{  \begin{array}{l}
\displaystyle{\dddot{q}_1 = \displaystyle{6 \frac{\ell_1\ell_2 (q_1-q_2) - \ell_1(\dot{q}_1-\dot{q}_2)}{(q_2-q_1)^4}}}\, ,\\
\displaystyle{\dddot{q}_2 = \displaystyle{6 \frac{\ell_1\ell_2 (q_2-q_1) - \ell_1(\dot{q}_2-\dot{q}_1)}{(q_2-q_1)^4}}} \, .\end{array}\right.
\end{equation}
will be called the third-order Calogero-Moser system.
\end{definition}
This system is nilpotent integrable because it has been obtained by reduction from the nilpotent integrable system $\dddot{X} = 0$.

Other reductions of the original system can be obtained by using a different tower of constants of the motion.  For instance, we may have used the `energy function'
$$
E = \frac12 \dot{X}^2 \, ,
$$
hence,
$$
\dot{E} = \dot{X} \cdot \ddot{X} \, ,\qquad \ddot{E} = \ddot{X}^2 \, , \qquad \dddot{E} = 0 \, .
$$
Let
$$
e_1 = \mathrm{Tr\,} (E) = \frac12 (\dot{q}_1^2 + \dot{q}_2^2) + q^2 \dot{\varphi} \, .
$$
Then, using the tower of functions $e_1$, $e_2 = \dot{e}_1$, $e_3 = \dot{e}_2$ ($\dot{e}_3 = 0$), we can obtain different reductions of the original system, all of them nilpotent integrable.

There are further reductions of the third-order Calogero-Moser system that finally will take us back to the second-order Calogero-Moser system and, in addition, become Hamiltonian.   The simplest way to do that is to use the constant of the motion $A = \ddot{X}$.   Denoting by $c_1 = \frac12 (A_{11}+ A_{22}) $, $c_2 = A_{12}$, and $c_3 =  \frac12 (A_{11}- A_{22})$, we deduce that
the Hamiltonian function
\begin{equation*}
H =\frac12(\dot{q}_1^2+ \dot{q}_2^2) +\frac{2\ell_1^2}{(q_1-q_2)^2}-\frac{c_1}{2}(q_1+q_2) + \frac12
(q_2-q_1) (-c_2\sin2\varphi+c_3 \cos 2 \varphi)
\end{equation*}
is a constant of the motion.  This Hamiltonian has been obtained from the obvious Hamiltonian for the system $\ddot{X} = A$, namely
$$
H = \frac12 (\dot{x}^2 + \dot{y}^2 + \dot{z}^2) -(c_1x+c_2y+c_3z),
$$
using the coordinates $q_1,q_2, \varphi$ as before and the angular constants of the motion.

However we cannot reduce this Hamiltonian further,  because it is not invariant with respect to the Abelian symmetry $\partial /\partial \varphi$, except when $c_2=c_3=0$.  In this case, the equations of motion Eq. \eqref{3rdCM} reduce to the simple form
\begin{equation}
\begin{cases}
\ddot{q}_1=& \displaystyle{\frac{2\ell_1^2}{(q_1-q_2)^3}}+ c_1,\\[15pt]
\ddot{q}_2=&-\displaystyle{\frac{2\ell_1^2}{(q_1-q_2)^3}}+ c_1.
\end{cases}
\end{equation}
which is just the Calogero-Moser system with a constant acceleration term.

\subsubsection{The fourth-order Calogero-Moser system}
We will conclude the discussion on the reduction of higher-order free systems by considering the angular reduction of the fourth-order system $X^{(4)} = 0$.
Proceeding again as in the third-order case, we get:
$$
X^{(4)} = G M^{(4)} G^{-1} \, ,\quad \mathrm{with} \quad M^{(4)} = \dot{M}^{3} + [\tau, M^{(3)}] \, .
$$
By using the structure of $M^{(3)}$ given by equation Eq. \eqref{M3}, we obtain
\begin{eqnarray*}
M^{(4)} &=& Q^{(4)} + \dot{M}^{(3)}_3 \sigma_3 + \dot{M}^{(3)}_1 \sigma_1 + [\tau,\dddot{Q}]+
M^{(3)}_3[\tau, \sigma_3] + M^{(3)}_1 [\tau,\sigma_1 ] \\
&=&  Q^{(4)} + M^{(4)}_3 \sigma_3 + M^{(4)}_1 \sigma_1
\end{eqnarray*}
with the recursive equations:
\begin{eqnarray*}
M^{(4)}_3 &=& \dot{M}^{(3)}_3+ 2\dot{\varphi}M^{(3)}_1  \, , \\
M^{(4)}_1 &=& \dot{\varphi}\dddot{q} + \dot{M}^{(4)}_1  - 2\dot{\varphi} M^{(4)}_3
\end{eqnarray*}
Thus a simple computation taking into account the explicit expressions for $M_3^{(3)}$ and $M_1^{(3)}$ given in Eqs. \eqref{M33}-\eqref{M31}, leads us to
\begin{eqnarray*}
M^{(4)} &=& Q^{(4)} + \left( \frac{6\ell_2^2 + 8\ell_1\ell_3}{q^3} - \frac{4\ell_1^2\ddot{q} + 32 \ell_1\ell_2 \dot{q}}{q^4} + \frac{8\ell_1^2\dot{q}^2}{q^5}- \frac{8\ell_1^4}{q^7}\right)  \sigma_3 \\
&+& \left(\frac{\ell_4}{q} + 2 \frac{\ell_1\dddot{q} - \ell_3\dot{q}}{q^2}  - 3\frac{\ell_1\dot{q}\ddot{q} - \ell_2\dot{q}^2}{q^3} - 14 \frac{\ell_1^2\ell_2}{q^5} - 18 \frac{\ell_1^3\dot{q}}{q^6}\right)  \sigma_1 \, .
\end{eqnarray*}
This gives us the fourth-order system
\begin{equation}
Q^{(4)} =  -\frac{6\ell_2^2 + 8\ell_1\ell_3}{q^3} + \frac{4\ell_1^2\ddot{q} + 32 \ell_1\ell_2 \dot{q}}{q^4} -  \frac{8\ell_1^2\dot{q}^2}{q^5} + \frac{8\ell_1^4}{q^7} \,.
\end{equation}

\section{Discussion and Future Perspectives}

In this work, we have tried to identify the main features which are needed to provide an explicit description of the flow of a given dynamical system and that do not rely on any additional geometrical structure.   The nilpotent algebra consisting of  the family of all higher-order constants of the motion, together with a maximal Abelian algebra of symmetries of the given dynamics, are sufficient to realize this description, provided that they satisfy a natural relation between them, i.e., that the maximal Abelian algebra contains the annihilator of the nilpotent algebra.    Natural examples, like higher-order free motion, that do not exhibit any obvious geometrical structure associated to them, fall in this category.
We have named these systems the nilpotent integrable ones, since their flows can be written as a polynomial in the time parameter.

Such systems generalize in a natural way the notion of Hamiltonian completely integrable systems. They exhibit canonical forms and the theory of reduction applies nicely to them. Also, we have shown how to construct examples of them without recurring to any additional geometrical structure like a symplectic or a Poisson one.  The algebraic language offered by the theory of differentiable algebras has proved to be particularly well adapted to make manifest the conceptual content of the notion of nilpotent integrability.   It is important to point out that our approach allows to treat dynamics on spaces exhibiting mild singularities.

To avoid to make this paper inordinately large, we do not have discussed how the general ideas exposed here particularize when we consider additional geometrical structures, or when such structures can be determined from the integrability properties of a given system.    Neither we have discussed the set theoretical counterpart of the theory, that is, we have not made any attempt to translate the results contained in this paper in the geometrical language of smooth manifolds, vector fields, etc.

Other issues of primary interest that can be treated in the present theoretical framework are for instance superintegrable systems and, in a different vein, the extension of the notion of integrability to non-commutative spaces.
These subjects will all be discussed in subsequent papers.

%\subsubsection{Flows on tori}

\end{document}